\documentclass[10pt,reqno,a4paper]{amsart}
\usepackage[british]{babel}
\usepackage{amssymb,amstext,amsthm,eucal,bbm,mathrsfs,amscd}
\usepackage{bbold}
\usepackage{stmaryrd}
\usepackage[hmargin=1in]{geometry}
\usepackage{array,color,slashed}
\usepackage[utf8]{inputenc}
\usepackage[small]{eulervm}
\usepackage{tgpagella}
\usepackage{enumitem,multicol,cite}
\usepackage[unicode]{hyperref}
\hypersetup{%
  pdftitle   = {Killing superalgebras for Lorentzian four-manifolds},
  pdfkeywords = {Spencer cohomology, Lie superalgebra, superPoincaré,
    Killing superalgebra, filtered deformations, rigid supersymmetry,
    supergravity},
  pdfauthor  = {Paul de Medeiros, José Figueroa-O'Farrill, Andrea Santi},
  pdfcreator = {\LaTeX\ with package \flqq hyperref\frqq}
}
\theoremstyle{plain}
\newtheorem{lemma}{Lemma}
\newtheorem{proposition}[lemma]{Proposition}
\newtheorem{theorem}[lemma]{Theorem}

\newtheorem{definition}[lemma]{Definition}
\theoremstyle{definition}
\newtheorem*{remark}{Remark}
\newcommand{\1}{\mathbb{1}}
\newcommand{\Hom}{\mathrm{Hom}}
\newcommand{\End}{\mathrm{End}}
\newcommand{\SO}{\mathrm{SO}}
\newcommand{\SL}{\mathrm{SL}}
\newcommand{\SU}{\mathrm{SU}}
\newcommand{\Spin}{\mathrm{Spin}}
\newcommand{\Cl}{C\ell}

\newcommand{\fsl}{\mathfrak{sl}}
\newcommand{\fso}{\mathfrak{so}}

\newcommand{\stab}{\mathfrak{stab}}

\newcommand{\fg}{\mathfrak{g}}
\newcommand{\fh}{\mathfrak{h}}
\newcommand{\fk}{\mathfrak{k}}
\newcommand{\fa}{\mathfrak{a}}
\newcommand{\fp}{\mathfrak{p}}

\newcommand{\fX}{\mathfrak{X}}
\newcommand{\RR}{\mathbb{R}}

\newcommand{\ZZ}{\mathbb{Z}}
\newcommand{\CC}{\mathbb{C}}
\newcommand{\eD}{\mathscr{D}}
\newcommand{\eE}{\mathscr{E}}
\newcommand{\eL}{\mathscr{L}}
\newcommand{\sbar}{\overline{s}}
\newcommand{\diag}{\mathrm{diag}}

\newcommand{\be}{\boldsymbol{e}}

\newcommand{\ev}{\operatorname{ev}}
\newcommand{\im}{\operatorname{im}}
\newcommand{\vol}{\operatorname{vol}}
\newcommand{\diver}{\operatorname{div}}

\newcommand{\one}{\omega^{(1)}}
\newcommand{\two}{\omega^{(2)}}
\newcommand{\ttwo}{\widetilde\omega^{(2)}}
\newcommand{\three}{\omega^{(3)}}

\DeclareMathOperator{\id}{Id}
\DeclareMathOperator{\Mat}{Mat}
\DeclareMathOperator{\AdS}{AdS}
\DeclareMathOperator{\Sph}{S}
\DeclareMathOperator{\Ric}{Ric}
\DeclareMathOperator{\NW}{NW}

\DeclareSymbolFont{greekletters}{OML}{cmr}{m}{it}
\DeclareMathSymbol{\varsigma}{\mathalpha}{greekletters}{"26}
\definecolor{orange}{rgb}{1.0, 0.5, 0.0}
\allowdisplaybreaks

\begin{document}
\title{Killing superalgebras for Lorentzian four-manifolds}
\author[de~Medeiros]{Paul de Medeiros}
\author[Figueroa-O'Farrill]{José Figueroa-O'Farrill}
\author[Santi]{Andrea Santi}
\address[PdM]{Department of Mathematics and Natural Sciences,
  University of Stavanger, 4036 Stavanger, Norway}
\address[JMF,AS]{Maxwell Institute and School of Mathematics, The University
  of Edinburgh, James Clerk Maxwell Building, Peter Guthrie Tait Road,
  Edinburgh EH9 3FD, Scotland, UK}
\thanks{EMPG-16-08}
\begin{abstract}
  We determine the Killing superalgebras underpinning field theories
  with rigid unextended supersymmetry on Lorentzian four-manifolds by
  re-interpreting them as filtered deformations of $\ZZ$-graded
  subalgebras with maximum odd dimension of the $N{=}1$ Poincaré
  superalgebra in four dimensions. Part of this calculation involves
  computing a Spencer cohomology group which, by analogy with a
  similar result in eleven dimensions, prescribes a notion of Killing
  spinor, which we identify with the defining condition for bosonic
  supersymmetric backgrounds of minimal off-shell supergravity in four
  dimensions. We prove that such Killing spinors always generate a Lie
  superalgebra, and that this Lie superalgebra is a filtered
  deformation of a subalgebra of the $N{=}1$ Poincaré superalgebra in
  four dimensions. Demanding the flatness of the connection defining
  the Killing spinors, we obtain equations satisfied by the maximally
  supersymmetric backgrounds. We solve these equations, arriving at
  the classification of maximally supersymmetric backgrounds whose
  associated Killing superalgebras are precisely the filtered
  deformations we classify in this paper.
\end{abstract}
\maketitle
\tableofcontents

\section{Introduction}
\label{sec:introduction}

A number of impressive exact results \cite{Pestun:2007rz,
  Kapustin:2009kz, Drukker:2010nc, Jafferis:2010un, Jafferis:2011zi,
  Kallen:2012cs, Hosomichi:2012he, Kallen:2012va, Kim:2012ava} obtained
in recent years via supersymmetric localisation have motivated a more
systematic exploration of quantum field theories with rigid
supersymmetry in curved space. A critical feature in many of these
calculations is the non-trivial rôle played by certain non-minimal
curvature couplings which regulate correlation functions, so a clear
understanding of the general nature of such couplings would be
extremely useful.

Several isolated examples of curved backgrounds which support rigid
supersymmetry, like spheres and anti-de~Sitter spaces (also various
products thereof), have been known for some time
\cite{Shuster:1999zf,Blau:2000xg}. Beyond these examples, the most
systematic strategy for identifying curved backgrounds which support
some amount of rigid supersymmetry has hereto been that pioneered by
Festuccia and Seiberg in \cite{Festuccia:2011ws}. In four dimensions,
they described how a large class of rigid supersymmetric non-linear
sigma-models in curved space can be obtained by taking a decoupling
limit (in which the Planck mass goes to infinity) of the corresponding
locally supersymmetric theory coupled to minimal off-shell
supergravity. In this limit, the gravity supermultiplet is effectively
frozen out, leaving only the fixed bosonic supergravity fields as data
encoding the geometry of the supersymmetric curved background.
Following this paradigm, several other works explored the structure of
rigid supersymmetry for field theories in various dimensions on curved
manifolds in both Euclidean and Lorentzian signature \cite{Jia:2011hw,
  Samtleben:2012gy, Klare:2012gn, Dumitrescu:2012ha, Cassani:2012ri,
  Liu:2012bi, deMedeiros:2012sb}.

A well-established feature of supersymmetric supergravity backgrounds
is that they possess an associated rigid Lie superalgebra
\cite{MR1642707,
  AFHS,GMT1,GMT2,PKT,JMFKilling,FOPFlux,NewIIB,ALOKilling,
  FigueroaO'Farrill:2004mx, FigueroaO'Farrill:2007ar,
  FigueroaO'Farrill:2007ic, FigueroaO'Farrill:2008ka,
  FigueroaO'Farrill:2008if, Kuzenko:2015lca} that we shall refer to as
the Killing superalgebra of the background. Indeed, with respect to an
appropriate superspace formalism, the construction described in
\cite[§6.4]{MR1642707} (and reviewed in
\cite{Kuzenko:2015lca}) explains how this Killing superalgebra may be
construed in terms of the infinitesimal rigid superisometries of a
given background supergeometry. The even part of the Killing
superalgebra contains the Killing vectors which preserve the
background, whereas the odd part is generated by the rigid
supersymmetries supported by the background. The image of the odd-odd
bracket for the Killing superalgebra spans a Lie subalgebra of Killing
vectors which preserve the background. This Lie subalgebra, together
with the rigid supersymmetries, generate an ideal of the Killing
superalgebra, which we call the Killing ideal of the background. The
utility of this construction is that it often allows one to infer
important geometrical properties of the background directly from the
rigid supersymmetry it supports. For example, in dimensions six, ten
and eleven, it was proved in \cite{FigueroaO'Farrill:2012fp,
  Figueroa-O'Farrill:2013aca} that any supersymmetric supergravity
background possessing more than half the maximal amount of
supersymmetry is necessarily (locally) homogeneous.

As a rule, the interactions in a non-linear theory with a local
(super)symmetry may be constructed unambiguously by applying the
familiar Noether procedure to the linearised version of the
theory. Indeed, this is the canonical method for deriving interacting
gauge theories in flat space, supergravity theories and their locally
supersymmetric couplings to field theory supermultiplets. However,
depending on the complexity of the theory in question, it may not be
the most wieldy technique and it is sometimes preferable to proceed
with some inspired guesswork, perhaps based on the assumption of a
particular kind of symmetry (e.g., conformal coupling in a conformal
field theory). Either way, the guiding principle is to deform (in some
sense) the free theory you know in the most general way that is
compatible with the symmetries you wish to preserve.

One way to motivate the construction we shall describe in this paper
is as an attempt to streamline the procedure for deducing which curved
backgrounds support rigid supersymmetry directly in terms of their
associated Killing superalgebras. Instead of applying the Noether
method to obtain some complicated local supergravity coupling, taking
a rigid limit, looking for supersymmetric backgrounds and then
computing the Killing superalgebras of those backgrounds, our strategy
will be to simply start with the unextended Poincaré superalgebra
(without R-symmetry) and obtain all the relevant Killing superalgebras
directly as filtered deformations (see below for the definition) of
its subalgebras. As expected for the deformation problem of an
algebraic structure, there is a cohomology theory which governs the
infinitesimal deformations. In this case this is a generalised Spencer
cohomology theory, studied in a similar context by Cheng and Kac in
\cite{MR1688484,MR2141501}. In the present work, we shall apply this
philosophy to the unextended Poincaré superalgebra on $\RR^{1,3}$,
following a similar analysis on $\RR^{1,10}$ pioneered in
\cite{Figueroa-O'Farrill:2015efc,Figueroa-O'Farrill:2015utu} which
yielded what might be considered a Lie-algebraic derivation of
eleven-dimensional supergravity.

Let us describe more precisely the problem we set out to solve.  Let
$(V,\eta)$ denote the Lorentzian vector space on which
four-dimensional Minkowski space is modelled, $\fso(V)$ the Lie
algebra of the Lorentz group and $S$ its spinor representation.  The associated
$N{=}1$ Poincaré superalgebra $\fp$ has underlying vector
space $\fso(V) \oplus S \oplus V$ and Lie brackets, for all $A,B\in
\fso(V)$, $s \in S$ and $v,w \in V$, given by
\begin{equation}
  \label{eq:poincare}
  [A,B] = AB - BA \qquad [A,s] = \sigma(A)s \qquad
  [A,v] = Av \qquad\text{and}\qquad [s,s] = \kappa(s,s)~,
\end{equation}
where $\sigma$ is the spinor representation of $\fso(V)$ and
$\kappa : \odot^2 S \to V$ is such that $\kappa(s,s) \in V$ is the
Dirac current of $s$.  (This and other relevant notions are defined in
the Appendix.)  The Poincaré superalgebra is $\ZZ$-graded by assigning
degrees $0$, $-1$ and $-2$ to $\fso(V)$, $S$ and $V$, respectively and
the $\ZZ_2$ grading is compatible with the $\ZZ$ grading, in that the
parity is the degree mod $2$.  More precisely, the even subalgebra is
the Poincaré algebra $\fp_{\bar 0} = \fso(V) \oplus V$ and the odd
subspace is $\fp_{\bar 1} = S$.  By a $\ZZ$-graded subalgebra $\fa$
of $\fp$ we mean a Lie subalgebra $\fa = \fa_0 \oplus \fa_{-1} \oplus
\fa_{-2}$, with $\fa_i \subset \fp_i$.

Now recall that a Lie superalgebra $\fg$ is said to be \emph{filtered}, if
it is admits a vector space filtration
\begin{equation*}
  \fg^\bullet~:\qquad \cdots \supset \fg^{-2} \supset \fg^{-1} \supset
  \fg^0 \supset \cdots~,
\end{equation*}
with $\cup_i \fg^i = \fg$ and $\cap_i \fg^i = 0$, which is compatible
with the Lie bracket in that $[\fg^i, \fg^j] \subset \fg^{i+j}$.
Associated canonically to every filtered Lie superalgebra
$\fg^\bullet$ there is a graded Lie superalgebra $\fg_\bullet =
\bigoplus_i \fg_i$, where $\fg_i = \fg^i/\fg^{i+1}$.
It follows from the fact that $\fg^\bullet$ is filtered that
$[\fg_i,\fg_j] \subset \fg_{i+j}$, hence $\fg_\bullet$ is graded.

We say that a Lie superalgebra $\fg$ is a \emph{filtered deformation} of
$\fa < \fp$ if it is filtered and its associated graded superalgebra
is isomorphic (as a graded Lie superalgebra) to $\fa$.  If we do not
wish to mention the subalgebra $\fa$ explicitly, we simply say that
$\fg$ is a \emph{filtered subdeformation} of $\fp$.

The problem we address in this note is the classification of filtered
subdeformations $\fg$ of $\fp$ for which $\fg_{-1} = S$ (and hence
$\fg_{-2} = V$).

This paper is organised as follows. In
Section~\ref{sec:spencer-cohomology} we define and calculate the
Spencer cohomology group $H^{2,2}(\fp_-,\fp)$ of the Poincaré
superalgebra. This is the main cohomological calculation upon which
the rest of our results are predicated. In particular we use it to
extract the equation satisfied by the Killing spinors, recovering in
this way the form of the (old minimal off-shell) supergravity Killing
spinor equation. We will also use this cohomological calculation as a
first step on which to bootstrap the calculation of infinitesimal
subdeformations of the Poincaré superalgebra. We give two proofs of
the main result in Section~\ref{sec:spencer-cohomology}
(Proposition~\ref{prop:killspin}): a traditional combinatorial proof
using gamma matrices and a representation-theoretic proof exploiting
the equivariance under $\fso(V)$. In
Section~\ref{sec:killing-superalgebra} we prove that the (minimal
off-shell) supergravity Killing spinors generate a Lie superalgebra,
and that this Lie superalgebra is a filtered subdeformation of
$\fp$. These results are contained in Theorem~\ref{thm:brackets} in
Section~\ref{sec:KSA} and Proposition~\ref{prop:KSAisFD} in
Section~\ref{sec:ksa-as-filtered}, respectively.  In
Section~\ref{sec:zero-curvature-eqns} we classify, up to local
isometry, the geometries admitting the maximum number of Killing
spinors. We do this by solving the zero curvature equations for the
connection relative to which the Killing spinors are parallel, and
this is done by first solving for the vanishing of the Clifford trace
of the curvature: this simplifies the calculation and might be of
independent interest.  Section~\ref{sec:maxim-supersymm-back} contains
the result of the classification of maximally supersymmetric
backgrounds up to local isometry:  apart from Minkowski space and
$\AdS_4$, we find the Lie groups admitting a Lorentzian bi-invariant
metric.  In Section~\ref{sec:maxim-fil-def} we finish the determination of
maximally supersymmetric filtered subdeformations of $\fp$ and recover
in this way the Killing superalgebras of the maximally supersymmetric
backgrounds found in Section~\ref{sec:maxim-supersymm-back}. In the
case of a Lie group with bi-invariant metric, we note that the Killing
ideal is a filtered deformation of $\fa=S\oplus V$ and also explicitly
describe all other associated maximally supersymmetric filtered
subdeformations of $\fp$.  The main result there is
Theorem~\ref{thm:final} in Section~\ref{sec:summary}.  Finally, in
Section~\ref{sec:conclusions}, we offer some conclusions.

Given the nature of this problem, it is inevitable that we shall
recover some known results and observations which it would be remiss
of us not to contextualise. In particular, in addition to $\RR^{1,3}$, our classification of
Killing superalgebras for maximally supersymmetric backgrounds yields,
up to local isometry, the following
conformally flat Lorentzian geometries:
\begin{itemize}
  \item $\AdS_4$;
  \item $\AdS_3 \times \RR$,  with $\AdS_3$ identified with $\SL(2,\RR)$
    with its bi-invariant metric;
  \item $\RR \times \Sph^3$, with $\Sph^3$ identified with $\SU(2)$ with its
    bi-invariant metric; and
  \item $\NW_4$, a symmetric plane wave isometric to the Nappi-Witten
    group with its bi-invariant metric.
\end{itemize}
We prove that the geometries above are indeed realised as the
maximally supersymmetric backgrounds of minimal off-shell supergravity
in four dimensions, in Lorentzian signature. That is, we do \emph{not}
assume the form of the supergravity Killing spinor equation
from the outset---we actually derive it via Spencer cohomology!  It
therefore follows that the first three geometries above are precisely
the maximally supersymmetric backgrounds obtained in
\cite{Festuccia:2011ws}. Indeed, the classification of maximally
supersymmetric backgrounds of minimal off-shell supergravity in four
dimensions has been discussed in various other contexts in the recent
literature, e.g., see \cite[§2.1]{Liu:2012bi},
\cite[§§4.2-3]{Kuzenko:2012vd}, \cite{Kuzenko:2015lca},
\cite[p.~2]{Butter:2015tra}, \cite[pp.12-13]{Kuzenko:2015rfx}. The
$\NW_4$ background is rarely mentioned explicitly---perhaps because,
unlike the other maximally supersymmetric Lorentzian backgrounds, it
has no counterpart in Euclidean signature---but it is noted in
\cite[p.~2]{Butter:2015tra} as a plane wave limit, albeit in the
context of $N=2$ supergravity backgrounds. It is also worth pointing
out that \cite[§2.1]{Liu:2012bi} contains several useful identities
(e.g., integrability conditions and covariant derivatives of Killing
spinor bilinears) that we also encounter in our construction of the
Killing superalgebra for minimal off-shell supergravity backgrounds.

\section{Spencer cohomology}
\label{sec:spencer-cohomology}

In this section we define and calculate the (even) Spencer cohomology
of the Poincaré superalgebra.  This calculation has two purposes.
The first is to serve as a first step in the classification of
filtered subdeformations of the Poincaré superalgebra which is
presented in Section~\ref{sec:maxim-fil-def}.  The second is to derive
the equation satisfied by the Killing spinors which, as we show in
Section~\ref{sec:killing-superalgebra}, generate the filtered
subdeformation.  The main result, whose proof takes
the bulk of the section, is Proposition~\ref{prop:killspin}.

\subsection{Preliminaries}

Let $\fp=\fp_{-2}\oplus\fp_{-1}\oplus\fp_{0}$, where $\fp_{-2}=V$,
$\fp_{-1}=S$ and $\fp_{0}=\fso(V)$, be the Poincaré superalgebra and
$\fp_{-}=\fp_{-2}\oplus\fp_{-1}$ the negatively graded part of $\fp$.
We will now determine some Spencer cohomology groups associated to
$\fp$. We recall that the cochains of the Spencer complex of $\fp$ are
linear maps $\wedge^p \fp_- \to \fp$ or, equivalently, elements of
$\wedge^p \fp_-^*\otimes\fp$, where $\wedge^ \bullet$ is meant here
in the super sense, and that the degree in $\fp$ is extended to the
space of cochains by declaring that $\fp_p^*$ has degree $-p$. The
spaces in the complexes of even cochains of small degree are given in
Table~\ref{tab:even-cochains-small}, although for $d=4$ there are
cochains also for $p=5,6$ which we omit.

\begin{table}[h!]
  \centering
  \caption{Even $p$-cochains of small degree}
  \label{tab:even-cochains-small}
  \begin{tabular}{c|*{5}{>{$}c<{$}|}}
    \multicolumn{1}{c|}{} & \multicolumn{5}{c}{$p$} \\\hline
    deg & 0 & 1 & 2 & 3 & 4 \\\hline
    0 & \fso(V) & \begin{tabular}{@{}>{$}c<{$}@{}} S \to S\\ V \to
                    V\end{tabular} & \odot^2 S \to V & & \\\hline
    2 & & V \to \fso(V) & \begin{tabular}{@{}>{$}c<{$}@{}}
                            \wedge^2 V \to V\\ V \otimes S \to S \\
                            \odot^2 S \to \fso(V) \end{tabular}
         & \begin{tabular}{@{}>{$}c<{$}@{}} \odot^3 S \to S\\
             \odot^2 S \otimes V \to V\end{tabular} & \odot^4 S \to V
        \\\hline
    4 & & & \wedge^2 V \to \fso(V) & \begin{tabular}{@{}>{$}c<{$}@{}}
                                        \odot^2 S \otimes V \to
                                        \fso(V) \\ \wedge^2 V \otimes
                                        S \to S \\ \wedge^3 V \to
                                        V \end{tabular} & \begin{tabular}{@{}>{$}c<{$}@{}}
                                        \odot^4 S \to \fso(V) \\
                                        \odot^3 S \otimes V \to S \end{tabular} \\\hline
  \end{tabular}
\end{table}

Let $C^{d,p}(\fp_-,\fp)$ be the space of $p$-cochains of degree
$d$.  The Spencer differential
\begin{equation*}
  \partial: C^{d,p}(\fp_-,\fp) \to C^{d,p+1}(\fp_-,\fp)
\end{equation*}
is the Chevalley--Eilenberg differential for the Lie superalgebra
$\fp_-$ relative to its module $\fp$ with respect to the adjoint
action. For $p=0,1,2$ and $d\equiv 0 \pmod 2$ it is explicitly given
by the following expressions:
\begin{align}
  \begin{split}\label{eq:Spencer0}
    &\partial : C^{d,0}(\fp_-,\fp) \to C^{d,1}(\fp_-,\fp)\\
    &\partial\zeta(X) = [X,\zeta]~,
  \end{split}
  \\
  \begin{split}\label{eq:Spencer1}
    &\partial : C^{d,1}(\fp_-,\fp)\to C^{d,2}(\fp_-,\fp)\\
    &\partial\zeta(X,Y) = [X,\zeta(Y)] - (-1)^{xy} [Y,\zeta(X)] - \zeta([X,Y])~,
  \end{split}
\\
  \begin{split}\label{eq:Spencer2}
    &\partial:C^{d,2}(\fp_-,\fp) \to C^{d,3}(\fp_-,\fp)\\
    &\partial\zeta(X,Y,Z) = [X,\zeta(Y,Z)]+(-1)^{x(y+z)}[Y,\zeta(Z,X)] + (-1)^{z(x+y)} [Z,\zeta(X,Y)] \\
    & {} \qquad\qquad\qquad - \zeta([X,Y],Z) - (-1)^{x(y+z)} \zeta([Y,Z],X) -(-1)^{z(x+y)} \zeta([Z,X],Y)~,
  \end{split}
\end{align}
where $x,y,\dots$ are the parity of elements $X,Y,\dots$ of $\fp_-$
and $\zeta\in C^{d,p}(\fp_-,\fp)$ with $p=0,1,2$ respectively. 

In this section we shall be interested in the groups
$H^{d,2}(\fp_-,\fp)$ with $d>0$ and even.  We first recall some basic
definitions. A $\mathbb Z$-graded Lie superalgebra
$\fa=\bigoplus\fa_p$ with negatively graded part
$\fa_-=\bigoplus_{p<0}\fa_p$ is called \emph{fundamental} if $\fa_-$ is
generated by $\fa_{-1}$ and \emph{transitive} if for any $X\in\fa_p$
with $p\geq 0$ the condition $[X,\fa_-]=0$ implies $X=0$.

\begin{lemma}
  \label{lem:p4}
  The Poincaré superalgebra $\fp=\fp_{-2}\oplus\fp_{-1}\oplus\fp_{0}$
  is fundamental and transitive. Moreover $H^{d,2}(\fp_-,\fp)=0$ for
  all even $d>2$.
\end{lemma}

\begin{proof}
  The first claim is a direct consequence of the fact that $\kappa(S,S)=V$
  and that the natural action of $\fso(V)$ on $V$ is faithful.  For
  any $\zeta\in C^{4,2}(\fp_-,\fp)=\Hom(\wedge^2 V,\fso(V))$ one has
  \begin{align*}
    &\partial\zeta(s_1,s_2,v_1)=-\zeta(\kappa(s_1,s_2),v_1)\\
    &\partial\zeta(v_1,v_2,s_1)=-\sigma( \zeta(v_1,v_2) ) s_1\\
    &\partial\zeta(v_1,v_2,v_3)=-\zeta(v_2,v_3)v_1-\zeta(v_3,v_1)v_2-\zeta(v_1,v_2)v_3
  \end{align*}
  where $s_1,s_2\in S$ and $v_1,v_2,v_3\in V$. The first equation
  implies $\operatorname{Ker}\partial|_{C^{4,2}(\fp_-,\fp)}=0$, since
  $\fp$ is fundamental, and therefore $H^{4,2}(\fp_-,\fp)=0$. Finally
  $C^{d,2}(\fp_-,\fp)=0$ and $H^{d,2}(\fp_-,\fp)=0$ for degree
  reasons, for all even $d>4$.
\end{proof}

Note that the space of cochains $C^{d,p}(\fp_-,\fp)$ is an
$\fso(V)$-module and the same is true for the spaces of cocycles and
coboundaries, as $\partial$ is $\fso(V)$-equivariant. This implies
that each cohomology group $H^{d,p}(\fp_-,\fp)$ is an
$\fso(V)$-module, in a natural way. It remains to compute
\begin{equation*}
  H^{2,2}(\fp_-,\fp) = \frac{\ker\partial: C^{2,2}(\fp_-,\fp) \to
    C^{2,3}(\fp_-,\fp)}{\partial C^{2,1}(\fp_-,\fp)}
\end{equation*}
and, in particular, to describe its $\fso(V)$-module structure. We consider the decomposition
\begin{equation*}
  C^{2,2}(\fp_-,\fp)=\Hom(\wedge^2 V, V)\oplus \Hom(V \otimes S, S)
\oplus \Hom(\odot^2S, \fso(V))
\end{equation*}
into the direct sum of $\fso(V)$-submodules and write any $\zeta\in
C^{2,2}(\fp_-,\fp)$ accordingly; i.e., $\zeta=\alpha+\beta+\gamma$ with
\begin{equation*}
  \begin{split}
    \alpha &\in \Hom(\wedge^2V,V)\phantom{ccccccccccccccc}\\
    \beta &\in\Hom(V \otimes S, S)\\
    \text{and}\qquad\gamma &\in \Hom(\odot^2S, \fso(V))\;.
  \end{split}
	\end{equation*}
We denote the associated $\fso(V)$-equivariant projections by
\begin{equation}
  \label{eq:projectors}
  \begin{split}
    \pi^{\alpha} &: C^{2,2}(\fp_-,\fp)\to \Hom(\wedge^2 V, V)\\
    \pi^{\beta} &: C^{2,2}(\fp_-,\fp)\to \Hom(V \otimes S, S)\\
    \text{and}\qquad\pi^{\gamma} &: C^{2,2}(\fp_-,\fp)\to \Hom(\odot^2S,
    \fso(V))\;.
  \end{split}
\end{equation}

\begin{lemma}
  \label{lem:iso}
  The component
  $\partial^\alpha=\pi^\alpha\circ\partial : \Hom(V,\fso(V))
  \longrightarrow \Hom(\wedge^2 V, V)$
  of the Spencer differential $\partial$ is an isomorphism. In
  particular,
  $\ker \partial|_{C^{2,2}(\fp_-,\fp)}=\partial
  \Hom(V,\fso(V))\oplus\mathscr H^{2,2}$,
  where $\mathscr H^{2,2}$ is the kernel of $\partial$ acting on
  $\Hom(V \otimes S, S) \oplus \Hom(\odot^2S, \fso(V))$, and every
  cohomology class $[\alpha + \beta + \gamma] \in H^{2,2}(\fp_-,\fp)$
  has a unique cocycle representative with $\alpha = 0$.
\end{lemma}

\begin{proof}
  The image of $\psi\in\Hom(V,\fso(V))$ under $\partial^\alpha$ is given by
  \begin{equation*}
    \partial^\alpha\psi(v_1,v_2) =
    \psi(v_1)v_2-\psi(v_2)v_1
  \end{equation*}
  where $v_1,v_2 \in V$ and the first claim of the lemma follows from
  classical arguments (see \cite{MR891190}; see also
  e.g., \cite{MR3056953, Figueroa-O'Farrill:2015efc}).
  
  Now for any given $\alpha \in \Hom(\wedge^2V,V)$, there is a unique
  $\psi \in \Hom(V,\fso(V))$ such that
  $\partial \psi = \alpha + \widetilde\beta + \widetilde\gamma$, for
  some $\widetilde\beta \in \Hom(V \otimes S, S)$ and
  $\widetilde\gamma \in \Hom(\odot^2S,\fso(V))$. Hence, given any
  cocycle $\zeta=\alpha + \beta + \gamma$, we may add the coboundary
  $\partial(-\psi)$ without changing its cohomology class and
  resulting in the cocycle
  $(\beta - \widetilde\beta) + (\gamma - \widetilde\gamma)$, which has
  no component in $\Hom(\wedge^2 V ,V)$. This proves the last claim
  of the lemma. The decomposition
  $ \ker\partial|_{C^{2,2}(\fp_-,\fp)}=\partial
  \Hom(V,\fso(V))\oplus\mathscr H^{2,2} $ is clear.
\end{proof}

\subsection{The cohomology group \texorpdfstring{$H^{2,2}(\fp_-,\fp)$}{H22(p-,p)}}

Lemma~\ref{lem:iso} gives a canonical identification
$H^{2,2}(\fp_-,\fp)\cong\mathscr H^{2,2}$ of
$\fso(V)$-modules. Furthermore it follows from
equation~\eqref{eq:Spencer2} that $\beta + \gamma$ is an element of
$\mathscr{H}^{2,2}$ if and only if the following pair of equations are
satisfied:
\begin{equation}
  \label{eq:cc1}
  \gamma(s,s)v = - 2 \kappa(s,\beta(v,s)) \qquad\forall\; s\in S, v\in V~,
\end{equation}
and
\begin{equation}
  \label{eq:cc2}
  \sigma( \gamma(s,s) ) s = -\beta(\kappa(s,s),s) \qquad \forall\; s\in S~.
\end{equation}
Note that \eqref{eq:cc1} fully expresses $\gamma$ in terms of $\beta$,
once the integrability condition that $\gamma$ takes values in
$\fso(V)$ has been taken into account. The solution of the
integrability condition and of equation \eqref{eq:cc2} is the content
of the following

\begin{proposition}
  \label{prop:killspin}
  Let   $\beta+\gamma\in \Hom(V\otimes S, S)
  \oplus\Hom(\odot^2S,\fso(V))$.  Then $\partial(\beta+\gamma)=0$ if
  and only if there exist $a,b\in\mathbb R$ and $\varphi\in V$ such
  that
  \begin{enumerate}[label=(\roman*)]
  \item $\beta(v,s)=v\cdot(a+b\vol)\cdot
    s-\frac{1}{2}(v\cdot\varphi+3\varphi\cdot v)\cdot\vol\cdot s$,
  \item $\gamma(s,s)v = - 2 \kappa(s,\beta(v,s))$,
  \end{enumerate}
  for all $v\in V$ and $s\in\ S$. In particular there is a canonical
  identification
  \begin{equation*}
    H^{2,2}(\fp_-,\fp)\simeq \mathscr H^{2,2}\simeq 2\mathbb R\oplus V
  \end{equation*}
  of $\fso(V)$-modules.
\end{proposition}

\begin{proof}
  We find it convenient to work relative to an $\eta$-orthonormal basis
  $(\be_\mu)$ for $V$.  In particular the formalism of
  Section~\ref{sec:gamma-matrices} is in force, as is the Einstein
  summation convention.

  Let us contract the cocycle condition \eqref{eq:cc1} with $w
  \in V$.  The left-hand side becomes
  \begin{equation}
    \eta(w,\gamma(s,s)(v)) = \gamma(s,s)_{\mu\nu}w^\mu v^\nu~,
  \end{equation}
  whereas the right-hand side becomes
  \begin{equation}
    -2 \eta(w,\kappa(s,\beta(v,s))) = -2 \left<s, w\cdot \beta(v,s)
      \right> = -2 w^\mu v^\nu \sbar \Gamma_\mu\beta_\nu s~,
  \end{equation}
  where we have introduced $\beta_\mu =
  \beta(\be_\mu,-)$.  In summary, the first cocycle condition becomes
  \begin{equation}
    w^\mu v^\nu \left(\gamma(s,s)_{\mu\nu} + 2 \sbar \Gamma_\mu
      \beta_\nu s\right) = 0~,
  \end{equation}
  which must hold for all $v,w \in V$, so that they can be abstracted
  to arrive at
  \begin{equation}
    \label{eq:cc2p}
    \gamma(s,s)_{\mu\nu} + 2 \sbar \Gamma_\mu \beta_\nu s = 0~.
  \end{equation}

  Symmetrising $(\mu\nu)$ we obtain the ``integrability condition''
  \begin{equation}
    \label{eq:cc3}
    \sbar \Gamma_{(\mu} \beta_{\nu)} s = 0~,
  \end{equation}
  whereas skew-symmetrising $[\mu\nu]$ and using that
  $\gamma(s,s)_{\mu\nu} = - \gamma(s,s)_{\nu\mu}$, we arrive at
  \begin{equation}
    \label{eq:cc4}
    \gamma(s,s)_{\mu\nu} = - 2 \sbar \Gamma_{[\mu} \beta_{\nu]}s~.
  \end{equation}
  Notice that, as advertised, this last equation simply expresses $\gamma$ in terms of
  $\beta$.  Acting on $s \in S$,
  \begin{equation}
    \begin{split}
      \sigma( \gamma(s,s) )s &= -\tfrac14 \gamma(s,s)_{\mu\nu} \Gamma^{\mu\nu}s \\
      &= \tfrac12 ( \sbar \Gamma_\mu \beta_\nu s ) \Gamma^{\mu\nu} s~,
    \end{split}
  \end{equation}
  and inserting this equation into the second cocycle condition~\eqref{eq:cc2}, we arrive at
  \begin{equation}
    \label{eq:cc5}
    ( \sbar \Gamma^\mu s ) \beta_\mu s + \tfrac12 ( \sbar\Gamma_\mu
    \beta_\nu s ) \Gamma^{\mu\nu}s = 0~.
  \end{equation}
  So we must solve equations~\eqref{eq:cc3} and \eqref{eq:cc5} for
  $\beta$.

  Since $\End(S) \cong \Cl(V) \cong \wedge^\bullet V$ (where the
  first isomorphism is one of algebras and the second one of vector
  spaces), we may write
  \begin{equation}
    \beta_\mu = \beta_\mu^{(0)} \1 + \beta_{\mu\nu}^{(1)}\Gamma^\nu +
    \tfrac12 \beta_{\mu\nu\rho}^{(2)} \Gamma^{\nu\rho} +
    \beta_{\mu\nu}^{(3)} \Gamma^\nu\Gamma_5 + \beta_\mu^{(4)}\Gamma_5~,
  \end{equation}
  with $\beta_\mu^{(i)}\in\wedge^i V$ so that
  \begin{equation}
    \label{eq:prelhscc3}
    \sbar \Gamma_\mu\beta_\nu s = \beta_\nu^{(0)} \sbar \Gamma_\mu s +
    \beta ^{(1)}_\nu{}^\rho \sbar \Gamma_{\mu\rho} s -
    \beta_{\nu\mu\rho}^{(2)} \sbar \Gamma^\rho s + \tfrac12
    \epsilon_{\mu\sigma\tau\rho} \beta^{(3)}_{\nu}{}^\rho
    \sbar \Gamma^{\sigma\tau} s~,
  \end{equation}
  where we have used the last of the duality
  equations~\eqref{eq:duality} and the symmetry
  relations~\eqref{eq:symmetry}.

  Inserting this into equation~\eqref{eq:cc3}, which must be true for all $s \in S$, we get that the terms which depend
  on $\sbar \Gamma^\rho s$ and $\sbar \Gamma^{\rho\sigma} s$ must vanish
  separately and we arrive at two equations:
  \begin{equation}
    \label{eq:cc3-1}
    \beta^{(0)}_\mu \eta_{\nu\rho} + \beta^{(0)}_\nu
    \eta_{\mu\rho} - \beta^{(2)}_{\mu\nu\rho} - 
    \beta^{(2)}_{\nu\mu\rho} = 0~,
  \end{equation}
  and
  \begin{equation}
    \label{eq:cc3-2}
    \eta_{\mu\rho}\beta^{(1)}_{\nu\sigma} + 
    \eta_{\nu\rho}\beta^{(1)}_{\mu\sigma} -
    \eta_{\mu\sigma}\beta^{(1)}_{\nu\rho} - 
    \eta_{\nu\sigma}\beta^{(1)}_{\mu\rho} +
    \beta^{(3)}_\mu{}^\tau \epsilon_{\nu\tau\rho\sigma} +
    \beta^{(3)}_\nu{}^\tau \epsilon_{\mu\tau\rho\sigma} = 0~.
  \end{equation}
  Tracing this last equation with $\eta^{\mu\nu}$, we learn that
  \begin{equation}
  \label{eq:cc3-2additional}
    \beta^{(1)}_{[\rho\sigma]} = - \tfrac12
    \beta^{(3)}{}^{\mu\nu} \epsilon_{\mu\nu\rho\sigma}~,
  \end{equation}
  whereas tracing \eqref{eq:cc3-2} with $\eta^{\nu\sigma}$ and using \eqref{eq:cc3-2additional}, results in
  \begin{equation}
    \beta^{(1)}_{\mu\rho} = a \eta_{\mu\rho} \qquad\text{and}\qquad
    \beta^{(3)}_{[\mu\rho]} = 0~,
  \end{equation}
  for $a = \tfrac14 \eta^{\mu\nu}\beta^{(1)}_{\mu\nu} \in \RR$.

  Substituting the expressions above back into equation~\eqref{eq:cc3-2}, we find
  \begin{equation}
    \beta^{(3)}_\mu{}^\tau \epsilon_{\nu\tau\rho\sigma} +
    \beta^{(3)}_\nu{}^\tau \epsilon_{\mu\tau\rho\sigma} = 0~.
  \end{equation}
  Multiplying by $\tfrac12 \epsilon^{\alpha\beta\rho\sigma}$, and using
  the identities \eqref{eq:epseps}, we obtain
  \begin{equation}
    -\delta_\mu^\alpha \beta^{(3)}_\nu{}^\beta + \delta_\mu^\beta
    \beta^{(3)}_\nu{}^\alpha - \delta_\nu^\alpha \beta^{(3)}_\mu{}^\beta
    + \delta_\nu^\beta \beta^{(3)}_\mu{}^\alpha = 0~.
  \end{equation}
  Tracing the expression above with $\eta^{\nu\beta}$, we arrive at
  \begin{equation}
    \beta^{(3)}_{\mu\alpha} = b \eta_{\mu\alpha}~,
  \end{equation}
  for $b = \tfrac14 \eta^{\mu\nu} \beta^{(3)}_{\mu\nu} \in \RR$.

  Tracing equation~\eqref{eq:cc3-1} with $\eta^{\mu\nu}$ gives
  \begin{equation}
    2 \beta^{(0)}_\rho - 2 \eta^{\mu\nu} \beta^{(2)}_{\mu\nu\rho} = 0~,
  \end{equation}
  while tracing it with $\eta^{\nu\rho}$ gives
  \begin{equation}
    5 \beta^{(0)}_\mu + \eta^{\nu\rho} \beta^{(2)}_{\nu\rho\mu} = 0~.
  \end{equation}
  These two equations together imply
  \begin{equation}
    \beta^{(0)}_\mu = 0~,
  \end{equation}
  which, when inserted into equation~\eqref{eq:cc3-1}, yields
  \begin{equation}
    \label{eq:presolcc3}
    \beta^{(2)}_{(\mu\nu)\rho} = 0~.
  \end{equation}
  This implies $\beta^{(2)}_{\mu\nu\rho} = \beta^{(2)}_{[\mu\nu\rho]}$ (i.e., $\beta^{(2)} \in \wedge^3V$ ), so
  that it can be parametrised by $\varphi \in V$ such that
  \begin{equation}
    \label{eq:beta2lambda3}
    \beta^{(2)}_{\mu\nu\rho} = \epsilon_{\mu\nu\rho\sigma} \varphi^\sigma~.
  \end{equation}
  In summary, the general solution of equation~\eqref{eq:cc3} is
  \begin{equation}
    \label{eq:solcc3}
    \beta_\mu = \Gamma_\mu  (a + b \Gamma_5) + \varphi^\nu
    \Gamma_{\mu\nu} \Gamma_5+
    \beta^{(4)}_\mu \Gamma_5~,
  \end{equation}
  where we have used the the last of the
  identities~\eqref{eq:duality}.

  Next we solve the second cocycle condition~\eqref{eq:cc5}.
  Using the expression for $\beta_\mu$ given in
  equation~\eqref{eq:solcc3}, we can rewrite the first term of
  equation~\eqref{eq:cc5} as follows:
  \begin{equation}
    ( \sbar \Gamma^\mu s ) \left( \Gamma_\mu (a+ b \Gamma_5) + \varphi^\nu
    \Gamma_{\mu\nu} \Gamma_5 + \beta^{(4)}_\mu \Gamma_5 \right) s,
  \end{equation}
  where, using that the Dirac current of $s$ Clifford annihilates $s$
  (see Proposition~\ref{prop:clifford}), the first term vanishes.
  Similarly, using $\Gamma_{\mu\nu} = - \Gamma_\nu \Gamma_\mu
  - \eta_{\mu\nu}$ and again the fact that $( \sbar\Gamma^\mu s )
  \Gamma_\mu s = 0$, the first term in equation~\eqref{eq:cc5} becomes
  \begin{equation}
    ( \sbar \Gamma^\mu s ) \left( \beta^{(4)}_\mu - \varphi_\mu \right)
    \Gamma_5 s~.
  \end{equation}
  We now rewrite the second term in equation~\eqref{eq:cc5} by
  inserting the expression for $\beta_\nu$ in
  equation~\eqref{eq:solcc3} into equation~\eqref{eq:prelhscc3} to
  obtain
  \begin{equation}
    \tfrac12 ( \sbar \Gamma_\mu \beta_\nu s ) \Gamma^{\mu\nu} s = \tfrac12
    ( \sbar \Gamma_{\mu\nu} (a + b \Gamma_5) s ) \Gamma^{\mu\nu} s - ( \sbar
    \Gamma^\mu s ) \varphi_\mu \Gamma_5 s~,
  \end{equation}
  where we have again used $\Gamma_{\mu\nu} = - \Gamma_\nu
  \Gamma_\mu - \eta_{\mu\nu}$ and the fact that $( \sbar\Gamma^\mu
  s ) \Gamma_\mu s = 0$.  The first term on the right-hand side vanishes
  by virtue of the fact that the Dirac $2$-form of $s$ and its dual both Clifford
  annihilate $s$ (see Proposition~\ref{prop:clifford}).  In
  summary, equation~\eqref{eq:cc5} becomes
  \begin{equation}
    ( \sbar \Gamma^\mu s ) \left( \beta^{(4)}_\mu - 2 \varphi_\mu \right)
    \Gamma_5 s  = 0~,
  \end{equation}
  for all $s \in S$, whose general solution is
  \begin{equation}
    \beta^{(4)}_\mu = 2 \varphi_\mu~.
  \end{equation}
  Inserting this into equation~\eqref{eq:solcc3}, we arrive at
  \begin{equation*}
    \beta_\mu = \Gamma_\mu  (a + b \Gamma_5) + \varphi^\nu
    \Gamma_{\mu\nu} \Gamma_5 + 2 \varphi_\mu \Gamma_5~,
  \end{equation*}
  which can be rewritten as
  \begin{equation*}
    \beta_\mu = \Gamma_\mu  (a + b \Gamma_5) - \tfrac12 \varphi^\nu
    \left( \Gamma_\mu \Gamma_\nu + 3 \Gamma_\nu
      \Gamma_\mu\right) \Gamma_5~,
  \end{equation*}
  from where the result follows.
\end{proof}

\begin{proof}[Alternative proof]
  It may benefit some readers to see an alternative proof of this result,
  which exploits the equivariance under $\fso(V)$.

  Let us consider the first cocycle condition~\eqref{eq:cc1}.
  Given $\beta\in\Hom(V,\End(S))$ and any $v\in V$ we let $\beta_v\in
  \End(S)$ to be defined by $\beta_v s = \beta(v,s)$ and rewrite
  \eqref{eq:cc1} as $\gamma(s,s)v = -2\kappa(s,\beta_v s)$.  Taking
  the inner product with $v$ and using \eqref{eq:DiracCurrent} and
  \eqref{eq:symplectic} we arrive at
  \begin{equation}
    \label{eq:symmetric-endo}
    0 = \left<s, v \cdot \beta_v s\right>\;,
  \end{equation}
  for all $s\in S$, $v\in V$.  In other words, for all $v\in V$, the
  endomorphism $v\cdot \beta_v$ of
  $S$ is in $\wedge^2 S=\wedge^0V\oplus\wedge^3 V\oplus\wedge^4 V$ 
  or, equivalently, it is fixed by the anti-involution $\varsigma$ defined by
  the symplectic form on $S$. We claim that the solution space of equation
  \eqref{eq:symmetric-endo} is an $\fso(V)$-submodule of
  $\Hom(V,\End(S))$. To see this, it is convenient to consider the
  $\fso(V)$-equivariant map
  \begin{equation*}
    \Upsilon: \Hom(V,\End(S)) \to \Hom(\odot^2V,\End(S))
  \end{equation*}
  which sends $\beta$ to $\Upsilon(\beta)$ given by
  \begin{equation*}
    \Upsilon(\beta)(v,w) = v\cdot \beta_w + w \cdot \beta_v~,
  \end{equation*}
  for all $v,w \in V$. We consider also the natural decompositions
  into $\fso(V)$-submodules
  \begin{equation}
    \begin{split}
      \label{eq:split}
      \Hom(V,\End(S)) &\cong \bigoplus_{p=0}^4 \Hom(V,\wedge^p V)\;,\\
      \Hom(\odot^2 V,\End(S)) &\cong \bigoplus_{q=0}^4\Hom(\odot^2V,\wedge^q V)\;,
    \end{split}
  \end{equation}
  which are induced by the usual identification
  $\End(S)=\bigoplus_{p=0}^4\wedge^p V$.  This allows us to write any
  elements $\beta\in\Hom(V,\End(S))$ and $\theta\in
  \Hom(\odot^2 V,\End(S))$ as $\beta=\beta_0+\cdots+\beta_4$ and
  $\theta= \theta_0 + \cdots +\theta_4$, where
  $\beta_p\in\Hom(V,\wedge^p V)$ and $\theta_q\in\Hom(\odot^2 V,
  \wedge^q V)$. The claim then follows from the fact that
  equation~\eqref{eq:symmetric-endo} is equivalent to $\Upsilon(\beta)_q =
  0$ for $q=1,2$.

  In Table~\ref{tab:betacomps} below we list the decomposition of
  $\Hom(V,\wedge^p V)$ for $p=0,1,\dots,4$ into irreducible
  $\fso(V)$-modules, with $(V \otimes \wedge^p V)_0$ denoting the
  kernel of Clifford multiplication
  $V \otimes \wedge^p V \to \wedge^{p-1}V \oplus \wedge^{p+1}V$.
  \begin{table}[h!]
    \centering
    \caption{Irreducible components of $\Hom(V,\wedge^pV)$ for
      $p=0,\dots,4$.}
    \label{tab:betacomps}
    \begin{tabular}{>{$}c<{$}|>{$}l<{$}}
      \multicolumn{1}{c|}{p} & \multicolumn{1}{c}{$\Hom(V,\wedge^p V)$}\\\hline
      0 & \wedge^1 V\\
      1 & \wedge^0 V \oplus \wedge^2 V \oplus (V\otimes\wedge^1 V)_0\\
      2 & 2\wedge^1 V \oplus (V \otimes
          \wedge^2V)_0\\
      3 & \wedge^0 V \oplus \wedge^2 V \oplus (V \otimes
          \wedge^1V)_0\\
      4 & \wedge^1 V
    \end{tabular}
  \end{table}

  From the first decomposition in \eqref{eq:split} we immediately infer that $\Hom(V,\End(S))$ is the direct sum of five 
  different isotypical components, namely 
  \begin{equation}
    \label{eq:isotypic1}
    2\wedge^0 V\;,\qquad 4\wedge^1 V\;,\qquad 2\wedge^2 V\;,
  \end{equation}
  and 
  \begin{equation}
    \label{eq:isotypic2}
    2(V\otimes \wedge^1 V)_0\;,\qquad  (V\otimes \wedge^2 V)_0\;.
  \end{equation}
  Note now that for any $\Theta, \Theta'\in\wedge^2 S$
  the element $\beta\in\Hom(V,\End(S))$ defined by
  \begin{equation*}
    \beta_v s=v\cdot\Theta\cdot s+\Theta'\cdot v\cdot s
  \end{equation*}
  satisfies 
  \begin{equation*}
    \begin{split}
      \varsigma(v\cdot\beta_v) &=
      -\eta(v,v)\varsigma(\Theta)+\varsigma(v\cdot\Theta'\cdot v)\\
      &= -\eta(v,v)\Theta+v\cdot\varsigma(\Theta')\cdot v\\
      &=-\eta(v,v)\Theta+v\cdot \Theta'\cdot v\\
      &=v\cdot\beta_v
    \end{split}
  \end{equation*}
  and it is therefore a solution of \eqref{eq:symmetric-endo}. If
  instead $\Theta, \Theta'\in\odot^2 S=\wedge^1V\oplus\wedge^2 V$ a
  similar computation yields
  $\varsigma(v\cdot\beta_v)=-v\cdot\beta_v$. In summary we get
  that the solution space of equation \eqref{eq:symmetric-endo}
  contains an $\fso(V)$-module isomorphic to
  \begin{equation}
    \label{eq:solmodule}
    \wedge^0 V \oplus 2 \wedge^3 V \oplus \wedge^4 V\;,
  \end{equation}
  where, say, $\Theta\in\wedge^0V\oplus\wedge^3 V\oplus\wedge^4 V$,
  $\Theta'\in\wedge^3 V$ and that there exists another submodule which is
  isomorphic to
  \begin{equation}
    \label{eq:notsolmodule}
    2\wedge^1 V\oplus2\wedge^2 V
  \end{equation}
  and formed by elements which do not satisfy
  \eqref{eq:symmetric-endo}. Note that the direct sum of
  \eqref{eq:solmodule} and  \eqref{eq:notsolmodule} gives all the
  isotypical components \eqref{eq:isotypic1} in $\Hom(V,\End(S))$.

  We now turn to the remaining isotypical components \eqref{eq:isotypic2}. 
  We first recall that $\Hom(V,\End(S))$ contains a single irreducible
  submodule of type $(V\otimes\wedge^2 V)_0$. We fix an orthonormal
  basis $(\be_{\mu})$ of $V$, consider the element
  \begin{equation*}
    \beta=\be_1^{\flat}\otimes \be_2\wedge \be_3+\be_2^\flat\otimes
    \be_1\wedge \be_3\in (V\otimes\wedge^2 V)_0
  \end{equation*}
  and evaluate
  \begin{equation*}
    \begin{split}
      \tfrac{1}{2}\Upsilon(\beta)(\be_1+\be_2,\be_1+\be_2) &= (\be_1+\be_2)\cdot \beta_{\be_1+\be_2}\\
      &=-(\be_1+\be_2)\cdot (\be_2\wedge \be_3)-(\be_1+\be_2)\cdot (\be_1\wedge \be_3)\\
      &=\imath_{\be_2}(\be_2\wedge \be_3)+\imath_{\be_1}(\be_1\wedge \be_3)\\
      &=-2 \be_3\;.
    \end{split}
  \end{equation*}
  In other words $\Upsilon(\beta)_1\neq 0$, which implies that
  $(V\otimes\wedge^2 V)_0$ is not included in the solution space of
  equation \eqref{eq:symmetric-endo}.  Finally any irreducible
  submodule in $\Hom(V,\End(S))$ isomorphic to $(V\otimes\wedge^1
  V)_0$ is given by the image into $\Hom(V,\wedge^1
  V)\oplus\Hom(V,\wedge^3 V)$ of an $\fso(V)$-equivariant embedding
  $\xi\mapsto (r_1\xi,r_2\xi)$,  $\xi\in(V\otimes\wedge^1 V)_0$,
  where $r_1,r_2\in\mathbb R$. For instance the image of
  $\xi=\be_1\otimes \be_2+ \be_2\otimes \be_1\in(V\otimes\wedge^1 V)_0$ is
  \begin{equation*}
    \begin{split}
      \beta &= r_1(\be_1^\flat\otimes \be_2+ \be_2^\flat\otimes
      \be_1)+r_2(\be_1^\flat \otimes \star \be_2+\be_2^\flat\otimes \star
      \be_1)\\
      &=r_1(\be_1^\flat\otimes \be_2+ \be_2^\flat\otimes \be_1)+r_2(-\be_1^\flat
      \otimes \be_0\wedge \be_1\wedge \be_3+\be_2^\flat\otimes \be_0\wedge
      \be_2\wedge \be_3)
    \end{split}
  \end{equation*}
  and we have
  \begin{equation*}
    \begin{split}
      \tfrac{1}{2}\Upsilon(\beta)(\be_1,\be_1)&=\be_1\cdot \beta_{\be_1}\\
      &=-r_1(\be_1\cdot \be_2)+r_2(\be_1\cdot \be_0\wedge \be_1\wedge \be_3 )\\
      &=-r_1 \be_1\wedge \be_2-r_2 \be_0\wedge \be_3\;.
    \end{split}
  \end{equation*}
  It follows that $\Upsilon(\beta)_2\neq 0$ unless $r_1=r_2=0$ and that
  the solution space of \eqref{eq:symmetric-endo} does not contain any
  submodule isomorphic to $(V\otimes\wedge^1 V)_0$ either.

  In summary we just showed that $\beta\in\Hom(V,\End(S))$ solves
  \eqref{eq:symmetric-endo} if and only if there exist reals $a,b$ and
  vectors $\varphi_1,\varphi_2$ such that
  \begin{equation}
    \label{eq:solfirstcoc}
    \beta_v s=v\cdot (a+b\vol)\cdot s+(v\cdot\varphi_1+\varphi_2\cdot
    v)\cdot\vol\cdot s\;,
  \end{equation}
  for all $v\in V$ and $s\in S$. 

  We now turn to equation \eqref{eq:cc2}, with $\beta$ as in
  \eqref{eq:solfirstcoc} and $\gamma$ expressed in terms of $\beta$
  using \eqref{eq:cc1}. We remark that from the above discussion we
  already know that $\mathscr H^{2,2}$ is identified with an
  $\fso(V)$-submodule of $2\wedge^0 V\oplus 2\wedge^1 V$.

  At this point it is convenient to fix an $\eta$-orthonormal basis
  $(\be_\mu)$ of $V$ and use the Einstein summation convention on
  indices as in Appendix~\ref{sec:gamma-matrices}.

  We first introduce
  \begin{equation*}
    \gamma(s,s)_{\mu\nu} = \eta(\be_\mu, \gamma(s,s)\be_\nu)
  \end{equation*}
  and note that \eqref{eq:cc1} is equivalent to
  $\gamma(s,s)_{\mu\nu}=-2\sbar\Gamma_\mu \beta_\nu s$ where we set
  $\beta_\mu = \beta_{\be_\mu}$. In particular,
  \begin{equation*}
    \begin{split}
      \sigma(\gamma(s,s))s&=-\tfrac{1}{4}\gamma(s,s)_{\mu\nu}\Gamma^{\mu\nu} s\\
      &=\tfrac{1}{2} ( \sbar \Gamma_\mu\beta_\nu s )\Gamma^{\mu\nu} s\;,\\
      \beta(\kappa(s,s),s)&= ( \sbar \Gamma^\mu s ) \beta_\mu s\;,
    \end{split}
  \end{equation*}
  and equation \eqref{eq:cc2} is equivalent to
  \begin{equation}
    \label{eq:cc2II}
    \tfrac{1}{2}( \sbar \Gamma_\mu\beta_\nu s) \Gamma^{\mu\nu} s + ( \sbar
    \Gamma^\mu s )\beta_\mu s=0\;.
  \end{equation}
  We first show that $\mathscr H^{2,2}$ includes the whole isotypical
  component $2\wedge^0 V$. Indeed if $\beta_v s=av\cdot s$ for some
  real $a$ then the left-hand side of equation \eqref{eq:cc2II} is $a(\frac{1}{2} ( \sbar
  \Gamma_{\mu\nu}s ) \Gamma^{\mu\nu} s + ( \sbar \Gamma^\mu s ) \Gamma_\mu s)$
  and both terms are zero separately since $\omega^{(2)}(s,s)\cdot
  s=\omega^{(1)}(s,s)\cdot s=0$ (see Proposition~\ref{prop:clifford}).
  If $\beta_v s=bv\cdot\vol\cdot s$, for some real $b$, we also get
  $b(\frac{1}{2}( \sbar \Gamma_{\mu\nu}\Gamma_5 s ) \Gamma^{\mu\nu} s +
  ( \sbar \Gamma^\mu s ) \Gamma_\mu \Gamma_5 s)=0$ since
  $\star\omega^{(2)}(s,s)\cdot s=\star\omega^{(1)}(s,s)\cdot s=0$ (see
  again Proposition~\ref{prop:clifford}).

  Finally, we consider the irreducible submodule in $2\wedge^1 V$
  determined by \eqref{eq:solfirstcoc} and the image of the
  $\fso(V)$-equivariant embedding
  $\varphi\mapsto(\varphi_1,\varphi_2)=(r_1\varphi,r_2\varphi)$, where
  $r_1, r_2\in\RR$.  In other words, we consider $\beta_v
  s=(r_1v\cdot\varphi+r_2\varphi\cdot v)\cdot\vol\cdot s$ and note
  that equation \eqref{eq:cc2II} gives
  \begin{equation}
    \label{eq:r1r2}
    \tfrac12 r_1 ( \sbar \Gamma_\mu \Gamma_\nu\varphi \Gamma_5 s ) \Gamma^{\mu\nu} s + 
    \tfrac12 r_2 ( \sbar \Gamma_\mu \varphi\Gamma_\nu \Gamma_5 s ) \Gamma^{\mu\nu} s +
    r_1 ( \sbar \Gamma^\mu s ) \Gamma_\mu \varphi \Gamma_5 s +
    r_2 ( \sbar \Gamma^\mu s ) \varphi\Gamma_\mu \Gamma_5 s=0~.
  \end{equation}
  The last term vanishes because $( \sbar \Gamma^\mu s ) \Gamma_\mu
  \Gamma_5 s =  - \omega^{(3)} \cdot s = 0$ (see
  Proposition~\ref{prop:clifford}).  The third term is
  \begin{equation*}
    r_1 \kappa \cdot \varphi \Gamma_5 s = -
    r_1 \varphi \cdot \kappa \Gamma_5 s - 2 r_1 \eta (\kappa,\varphi)
    \Gamma_5 s = - 2 r_1 \eta (\kappa,\varphi) \Gamma_5 s~,
  \end{equation*}
  again, using that $\omega^{(3)} \cdot s = 0$.  Using
  equation~\eqref{eq:duality} repeatedly, the Clifford relation and
  Proposition~\ref{prop:clifford} again, we can rewrite the first two
  terms of \eqref{eq:r1r2} as
  \begin{equation*}
    - r_1 \eta (\kappa,\varphi) \Gamma_5 s + r_2 \eta (\kappa,\varphi)\Gamma_5 s~,
  \end{equation*}
  turning equation~\eqref{eq:r1r2} into
  \begin{equation*}
    (r_2 - 3 r_1 )  \eta (\kappa,\varphi) \Gamma_5 s =0~.
  \end{equation*}
  Since this must hold for all $\varphi \in \wedge^1 V$ and $s \in S$, it follows that $r_2 = 3r_1$.
\end{proof}

\section{Killing superalgebras}
\label{sec:killing-superalgebra}

In analogy with the results
\cite{Figueroa-O'Farrill:2015efc,Figueroa-O'Farrill:2015utu} in eleven
dimensions, we define a notion of Killing spinor from the component
$\beta$ of the cocycle in Proposition~\ref{prop:killspin}. In this
section we prove that these Killing spinors generate a Lie
superalgebra.

\subsection{Preliminaries}
	
Let $(M,g,a,b,\varphi)$ be a four-dimensional Lorentzian spin manifold
$(M,g)$ with spin bundle $S(M)$ which is, in addition, endowed with two
functions $a,b\in\mathscr C^{\infty}(M)$ and a vector field
$\varphi\in\fX(M)$.  The main aim of this section is to construct
a Lie superalgebra $\mathfrak{k}=\mathfrak{k}_{\bar
  0}\oplus\mathfrak{k}_{\bar 1}$ naturally associated with
$(M,g,a,b,\varphi)$.
	
Motivated by (i) of Proposition~\ref{prop:killspin} we introduce the
connection
\begin{equation}
  \label{eq:gravitinoconnection}
  D_X\varepsilon:=\nabla_{X}\varepsilon-X\cdot (a+b\vol)\cdot \varepsilon+(\varphi\wedge X)\cdot\vol\cdot\varepsilon-2g(\varphi,X)\vol\cdot\varepsilon
\end{equation}
on $S(M)$, where $\nabla$ is the Levi-Civita connection of $(M,g)$,
$X\in\mathfrak{X}(M)$, $\varepsilon\in\Gamma(S(M))$.

\begin{definition}
  A section $\varepsilon$ of $S(M)$ is called a \emph{Killing spinor} if
  $D_X\varepsilon=0$ for all $X\in\mathfrak{X}(M)$.
\end{definition}

Note that any non-zero Killing spinor is nowhere vanishing since it is
parallel with respect to a connection on the spinor bundle. We set
\begin{equation}
  \label{eq:KSA}
  \begin{split}
    \mathfrak{k}_{\bar 0}&=\{X\in\mathfrak{X}(M)\mid\mathscr
    L_{X}g=\mathscr L_{X} a=\mathscr L_{X}b=\mathscr L_{X}\varphi=0\}\;,\\
    \mathfrak{k}_{\bar 1}&=\{\varepsilon\in\Gamma(S(M))\mid
    D_X\varepsilon=0\;\;\text{for all}\;\;X\in\mathfrak{X}(M)\}\;,
  \end{split}
\end{equation}
and consider the operation $[-,-]:\fk\otimes\fk\to\fk$  compatible
with the parity of $\mathfrak k=\mathfrak k_{\bar 0}\oplus\mathfrak
k_{\bar 1}$ and determined by the following maps:
\begin{itemize}
\item $[-,-]:\fk_{\bar 0}\otimes\fk_{\bar 0}\to\fk_{\bar 0}$ is given
  by the usual commutator of vector fields,
\item $[-,-]:\fk_{\bar 1}\otimes\fk_{\bar 1}\to\fk_{\bar 0}$ is a
  symmetric map, with
  $[\varepsilon,\varepsilon]=\kappa(\varepsilon,\varepsilon)$ given by
  the Dirac current of $\varepsilon\in\fk_{\bar 1}$,
\item $[-,-]:\fk_{\bar 0}\otimes\fk_{\bar 1}\to\fk_{\bar 1}$ is given
  by the spinorial Lie derivative of Lichnerowicz and Kosmann (see
  \cite{MR0312413} and also, e.g., \cite{FigueroaO'Farrill:1999va}).
\end{itemize}
The fact that $[-,-]$ actually takes values in $\fk$ is a consequence
of Theorem~\ref{thm:brackets} below, where we show that $[-,-]$ is
the bracket of a Lie superalgebra structure on $\fk$.  Assuming that
result for the moment we make the following

\begin{definition}
  The pair $(\fk=\fk_{\bar 0}\oplus\fk_{\bar 1},[-,-])$ is called the
  \emph{Killing superalgebra} associated with  $(M,g,a,b,\varphi)$.
\end{definition}

We recall that the spinorial Lie derivative of a spinor field
$\varepsilon$ along a Killing vector field $X$ is defined by
$\mathscr L_{X}\varepsilon = \nabla_{X}\varepsilon +
\sigma(A_X)\varepsilon$,
where $\sigma:\fso(TM)\to\End(S(M))$ is the spin representation and
$A_X=-\nabla X\in\fso(TM)$.  It enjoys the following basic properties,
for all Killing vectors $X,Y$, spinors $\varepsilon$, functions $f$
and vector fields $Z$:
\begin{enumerate}[label=(\roman*)]
\item $\mathscr L_X$ is a derivation:
  \begin{equation*}
    \mathscr L_X(f\varepsilon)=X(f)\varepsilon+f\mathscr L_{X}\varepsilon\;;
  \end{equation*}
\item $X\mapsto\mathscr L_X$ is a representation of the Lie algebra of Killing vector fields:
  \begin{equation*}
    \mathscr L_X(\mathscr L_Y\varepsilon)-\mathscr L_Y(\mathscr L_X\varepsilon)=\mathscr L_{[X,Y]}\varepsilon\;;
  \end{equation*}
\item $\mathscr L_X$ is compatible with Clifford multiplication:
  \begin{equation*}
    \mathscr L_X(Z\cdot \varepsilon)=[X,Z]\cdot\varepsilon+ Z\cdot\mathscr L_{X}\varepsilon\;;
  \end{equation*}
\item $\mathscr L_X$ is compatible with the Levi-Civita connection:
  \begin{equation*}
    \mathscr L_X(\nabla_{Z}\varepsilon)=\nabla_{[X,Z]}\varepsilon+ \nabla_{Z}(\mathscr L_{X}\varepsilon)\;.
  \end{equation*}
\end{enumerate}
We note for later use that, from property (iii) and the fact that
$\odot^2 S=\wedge^1 V\oplus\wedge^2 V$, we have for any Killing
vector $X$, spinor $\varepsilon$ and vector field $Z$,
\begin{equation*}
  \begin{split}
    g([X,\kappa(\varepsilon,\varepsilon)],Z) &=
    X(g(\kappa(\varepsilon,\varepsilon),Z)) - g(\kappa(\varepsilon,\varepsilon),[X,Z])\\
    &=X(\left<\varepsilon,Z\cdot\varepsilon\right>)-\left<\varepsilon,[X,Z]\cdot\varepsilon\right>\\
    &=2\left<\nabla_X\varepsilon,Z\cdot\varepsilon\right>+\left<\varepsilon,\nabla_{Z}X\cdot\varepsilon\right>\\
    &=2\left<\nabla_X\varepsilon,Z\cdot\varepsilon\right>+2\left<\varepsilon,Z\cdot\sigma(A_X)\varepsilon\right>\\
    &=2g(\kappa(\mathscr L_{X}\varepsilon,\varepsilon),Z)\;,
  \end{split}
\end{equation*}
which yields the following additional property of the spinorial Lie derivative:
\begin{itemize}
\item[(v)] the Dirac current is equivariant under the action of Killing vector fields:
\begin{equation*}
  [X,\kappa(\varepsilon,\varepsilon)]=2\kappa(\mathscr L_{X}\varepsilon,\varepsilon)\;.
\end{equation*}
\end{itemize}

We first collect a series of important auxiliary results, which will
be needed in the proof of the main Theorem~\ref{thm:brackets}. 

\begin{proposition}
  \label{prop:auxiliary}
  Let $\varepsilon$ be a non-zero section of the spinor bundle $S(M)$ of
  $(M,g,a,b,\varphi)$, with associated differential forms
  \begin{itemize}
  \item $\omega^{(1)}\in\Omega^1(M)$, where $\omega^{(1)}(X)=
    \left<\varepsilon, X\cdot \varepsilon\right>$,
  \item $\omega^{(2)}\in\Omega^2(M)$, where $\omega^{(2)}(X,Y)=
    \left<\varepsilon, (X \wedge Y) \cdot \varepsilon\right>$,
  \item $\widetilde\omega^{(2)}=-\star\omega^{(2)}\in\Omega^2(M)$, where
    $\widetilde\omega^{(2)}(X,Y)=  \left<\varepsilon, (X\wedge
      Y)\cdot\vol\cdot \varepsilon \right>$,
  \item $\omega^{(3)}=-\star\omega^{(1)}\in\Omega^3(M)$, where
    $\omega^{(3)}(X,Y,Z)=\left<\varepsilon, (X\wedge Y\wedge Z)\cdot\vol\cdot\varepsilon\right>$,
  \end{itemize}
  for all $X,Y,Z\in\mathfrak{X}(M)$. If $\varepsilon$ is a Killing spinor then
  \begin{enumerate}[label=(\roman*)]
  \item $d\omega^{(1)}=-4a\omega^{(2)}-4b\widetilde\omega^{(2)}-4\imath_{\varphi}\omega^{(3)}$,
  \item $d\omega^{(2)}=6b\omega^{(3)}$,
  \item $d\widetilde\omega^{(2)}=-6a\omega^{(3)}$,
  \item $d\omega^{(3)}=0$.
  \end{enumerate}
  In particular the Dirac current $K=\kappa(\varepsilon,\varepsilon)$ of
  $\varepsilon$ is a Killing vector field satisfying
  \begin{equation}
    \label{eq:Lderivatives}
    \mathscr L_{K}a = \mathscr L_{K} b = \mathscr L_{K}\one= \mathscr
    L_{K}\two = \mathscr L_{K} \ttwo = \mathscr L_{K} \three=0\;
  \end{equation}
  and
  \begin{equation}
    \label{eq:lienabla}
    \begin{split}
      0 &= -2\ttwo(Z,X) g(\mathscr L_{K}\varphi,Y)+2\ttwo(Z,Y) g(\mathscr
      L_{K}\varphi,X) -2\ttwo(\mathscr L_{K}\varphi,Y)g(Z,X)\\
      & \quad {} + 2\ttwo(\mathscr L_{K}\varphi,X)g(Z,Y) +4\ttwo(X,Y) g(\mathscr
      L_{K}\varphi,Z)\;,
    \end{split}
  \end{equation}
  for all $X,Y,Z\in\mathfrak{X}(M)$.
\end{proposition}

\begin{proof}
For any Killing spinor $\varepsilon$ and $X,Y,Z\in\mathfrak{X}(M)$ we compute
\begin{equation*}
  \begin{split}
    (\nabla_{Z}\omega^{(1)})(X)&=2\left<\varepsilon,X\cdot\nabla_{Z}\varepsilon\right>\\
    &=2a\left<\varepsilon, X\wedge Z\cdot\varepsilon \right> +2b\left<\varepsilon, X\wedge
      Z\cdot\vol\cdot \varepsilon\right> +2\left<\varepsilon, \varphi\wedge X\wedge
      Z\cdot\vol\cdot\varepsilon\right>\;,
  \end{split}
\end{equation*}
and
\begin{equation*}
  \begin{split}
    (\nabla_{Z}\omega^{(2)})(X,Y)&=2\left<\varepsilon,X\wedge Y\cdot\nabla_{Z}\varepsilon\right>\\
    &=2a\left<\varepsilon,X\wedge Y\cdot Z\cdot\varepsilon \right>
    +2b\left<\varepsilon,X\wedge Y\cdot Z\cdot\vol\cdot\varepsilon     \right>\\
    &\;\;\;\;-2\left<\varepsilon,X\wedge Y\cdot\varphi\wedge Z\cdot\vol\cdot\varepsilon
    \right>
    +4g(\varphi,Z)\left<\varepsilon,X\wedge Y\cdot\vol\cdot\varepsilon     \right>\\
    &=2ag(Z,X)\left<\varepsilon,Y\cdot\varepsilon \right>-2ag(Z,Y)\left<\varepsilon,X\cdot\varepsilon
    \right>
    +2b\left<\varepsilon,Z\wedge X\wedge Y\cdot\vol\cdot\varepsilon     \right>\\
    &\;\;\;\;+2g(\varphi,Y)\left<\varepsilon,X\wedge Z\cdot\vol\cdot\varepsilon \right>
    -2g(\varphi,X)\left<\varepsilon,Y\wedge Z\cdot\vol\cdot\varepsilon       \right>\\
    &\;\;\;\;-2g(Z,Y)\left<\varepsilon,X\wedge \varphi\cdot\vol\cdot\varepsilon \right>
    +2g(X,Z)\left<\varepsilon,Y\wedge\varphi\cdot\vol\cdot\varepsilon      \right>\\
    &\;\;\;\;+4g(\varphi,Z)\left<\varepsilon,X\wedge Y\cdot\vol\cdot\varepsilon \right>\;,
  \end{split}
\end{equation*}
where, in both cases, the last equality follows from
equation~\eqref{eq:symmetry} or, equivalently, that $\odot^2 S=\wedge^1
V\oplus\wedge^2 V$. In other words we have
\begin{align}
  \label{eq:nabla1}
  \nabla_{Z}\one&=-2a\imath_{Z}\two-2b\imath_{Z}\ttwo-2\imath_{Z}\imath_{\varphi}\three\;,\\
  \label{eq:nabla2}
  \nabla_{Z}\two&=2aZ\wedge\one+2b\imath_{Z}\three
                  -2\imath_Z\ttwo\wedge\varphi-2Z\wedge\imath_{\varphi}\ttwo+4g(\varphi,Z)\ttwo\;,
\end{align}
and applying $\star$, which is a parallel endomorphism of $\Omega^\bullet(M)$,
on both sides of these identities we also get
\begin{align}
  \nabla_{Z}\ttwo&=2b Z\wedge\one-2a\imath_{Z}\three
                   +2 Z\wedge\imath_{\varphi}\two+ 2\imath_{Z}\two\wedge\varphi-4g(\varphi,Z)\two\;,\\
  \nabla_{Z}\three&=2aZ\wedge\ttwo -2bZ\wedge\two +2 Z\wedge\varphi\wedge\one\;.
\end{align}
Claims (i)-(iv) follows then immediately from the fact that for any
$\omega\in\Omega^p(M)$ we have
\begin{equation*}
  d\omega=\sum_{\mu=0}^{3}e^\mu\wedge\nabla_{e_\mu}\omega
  \qquad\qquad\text{and}\qquad\qquad
  \sum_{\mu=0}^{3}e^\mu\wedge\imath_{e_\mu}\omega=p\omega\;,
\end{equation*}
where $(e_\mu)$ is a fixed local orthonormal frame field of $(M,g)$.

Now, for any Killing spinor $\varepsilon$ and $X,Y\in\mathfrak{X}(M)$ we have
\begin{equation*}
  \begin{split}
    g(\nabla_{X}K,Y)&=
    2\left<\varepsilon,Y\cdot\nabla_X\varepsilon\right>\\
    &=2a\left<\varepsilon,Y\cdot
      X\cdot\varepsilon\right>+2b\left<\varepsilon,Y\cdot
      X\cdot\vol\cdot\varepsilon\right>-2\left<\varepsilon,Y\cdot (\varphi\wedge X)\cdot\vol\cdot\varepsilon\right> +4g(\varphi,X)\left<\varepsilon,Y\cdot \vol\cdot\varepsilon\right>\\
    &=2a\left<\varepsilon,(Y\wedge
      X)\cdot\varepsilon\right>+2b\left<\varepsilon,(Y\wedge
      X)\cdot\vol\cdot\varepsilon\right>+2\left<\varepsilon,
      (\varphi\wedge Y\wedge X)\cdot\vol\cdot\varepsilon\right>
  \end{split}
\end{equation*}
where the last equality follows from
equation~\eqref{eq:symmetry}. Since the last term is manifestly
skewsymmetric in $X$ and $Y$ we have that $K$ is a Killing vector.  From $d\three=0$, we also have
\begin{equation*}
  0 = d(d\two) = 6db\wedge \three = -6db\wedge \star\one = -6(\imath_{K}db)\vol\;;
\end{equation*}
i.e., $\mathscr L_{K}b=0$. One shows $\mathscr L_{K}a=0$ in a similar
way. If $\omega=\one, \two, \ttwo,$ or $\three$, then
$\imath _{K}\omega=0$ by Proposition~\ref{prop:clifford} and from
(i)-(iv) we get
\begin{equation*}
  \mathscr L_{K}\omega = d\imath_{K}\omega+\imath_{K}d\omega =0\;.
\end{equation*}
This proof of \eqref{eq:Lderivatives} is thus completed.

In order to show \eqref{eq:lienabla} we use that $K$ is a Killing
vector and $\mathscr L_{K}\two=0$ so that for all
$X,Y,Z\in\mathfrak{X}(M)$:
\begin{equation*}
  \begin{split}
    0&=(\mathscr L_{K}\nabla_Z\two)(X,Y)-(\nabla_{[K,Z]}\two)(X,Y)\\
    &=\mathscr L_{K}((\nabla_Z\two)(X,Y))-(\nabla_{[K,Z]}\two)(X,Y) - \nabla_Z\two([K,X],Y)-\nabla_Z\two(X,[K,Y])\\
    &=-2\ttwo(Z,X)g(\mathscr L_{K}\varphi,Y)+2\ttwo(Z,Y) g(\mathscr
    L_{K}\varphi,X) -2\ttwo(\mathscr L_{K}\varphi,Y)g(Z,X)\\
    & \qquad {} + 2\ttwo(\mathscr L_{K}\varphi,X)g(Z,Y) + 4\ttwo(X,Y)
    g(\mathscr L_{K}\varphi,Z)\;,
  \end{split}
\end{equation*}
where the last identity follows from a direct computation using
\eqref{eq:nabla2} and \eqref{eq:Lderivatives}.
\end{proof}

\subsection{The Killing superalgebra}
\label{sec:KSA}

We state and prove the main result of Section~\ref{sec:killing-superalgebra}.

\begin{theorem}
  \label{thm:brackets}
  Let $X,Y\in\fk_{\bar 0}$ and $\varepsilon\in\fk_{\bar 1}$. Then
  $[X,Y]\in\fk_{\bar 0}$, $\kappa(\varepsilon,\varepsilon)\in\fk_{\bar
    0}$ whereas $\mathscr L_{X}\varepsilon\in\fk_{\bar 1}$. Moreover,
  $[-,-]$ defines a Lie superalgebra on $\fk=\fk_{\bar
    0}\oplus\fk_{\bar 1}$.
\end{theorem}

\begin{proof}
  The fact that $[\fk_{\bar 0},\fk_{\bar 0}]\subset\fk_{\bar 0}$
  follows from basic properties of Lie derivatives of vector
  fields. On the other hand for any $X\in\fk_{\bar 0}$ and
  $Z\in\mathfrak{X}(M)$ we have that
  \begin{equation*}
    [\eL_X, D_Z] = D_{[X,Z]}~,
  \end{equation*}
  since $D$ depends solely on the data $(g,a,b,\varphi)$ which is
  preserved by $X \in \fk_{\bar 0}$.
This shows that $\mathscr L_{X}\varepsilon$ is a Killing spinor or, in
other words, that $[\fk_{\bar 0},\fk_{\bar 1}]\subset\fk_{\bar 1}$.

We already know from Proposition~\ref{prop:auxiliary} that
$K=\kappa(\varepsilon,\varepsilon)$ is a Killing vector field which
satisfies $\mathscr L_{K}a=\mathscr L_{K}b=0$. To prove $K\in\fk_{\bar
  0}$ we still need to show $\mathscr L_{K}\varphi=0$.  From
Proposition~\ref{prop:auxiliary} we have
\begin{equation*}
  \begin{split}
    0&=-\frac{1}{4}d(d\one)=da\wedge\two +6ab\three+db\wedge\ttwo-6ab\three+d\imath_{\varphi}\three\\
    &=da\wedge\two+db\wedge\ttwo-\mathscr L_{\varphi}\star\one
  \end{split}
\end{equation*}
and hence, for any $\vartheta\in\Omega^1(M)$,
\begin{equation}
\label{eq:liestarII}
\begin{split}
\vartheta\wedge\star \mathscr L_{\varphi}\one&=\vartheta\wedge\mathscr L_{\varphi}\star\one-\diver(\varphi)\vartheta\wedge\star\one-(\mathscr L_{\varphi}g)(\vartheta,\one)\vol\\
  &=\vartheta\wedge da\wedge\two+\vartheta\wedge
  db\wedge\ttwo+\diver(\varphi)\vartheta\wedge\three -(\mathscr
  L_{\varphi}g)(\vartheta,\one)\vol\;.
\end{split}
\end{equation}
In the special case where $\imath_{K}\vartheta=0$ the first three
terms of the right-hand side of the above identity are degenerate $4$-forms and
hence zero. Then equation~\eqref{eq:liestarII} becomes
\begin{equation*}
  \begin{split}
    0 &= \vartheta\wedge\star \mathscr L_{\varphi}\one + (\mathscr
    L_{\varphi}g)(\vartheta,\one)\vol\\
    &= -g(\mathscr L_{\varphi}\vartheta,\one)\vol\\
    &= -(\mathscr L_{\varphi}\vartheta)(K)\vol\\
    &= \vartheta(\mathscr L_{\varphi}K)\vol\\
    &= -\vartheta(\mathscr L_{K}\varphi)\vol~,
  \end{split}
\end{equation*}
so that $\mathscr L_{K}\varphi=f K$, for some
$f\in\mathscr C^{\infty}(M)$. From this fact, equation
\eqref{eq:lienabla} and $\one\wedge\ttwo=0$ we finally get
\begin{align*}
0&=f(2\ttwo(X,Y)\one(Z)+\ttwo(X,Z)\one(Y)+\ttwo(Z,Y)\one(X))\\
&=3f \ttwo(X,Y)\one(Z)\;,
\end{align*}
for all $X,Y,Z\in\mathfrak{X}(M)$, hence $f=0$. This proves $\mathscr
L_{K}\varphi=0$ and $[\fk_{\bar 1},\fk_{\bar 1}]\subset\fk_{\bar 0}$.

We finally show that $[-,-]:\fk\otimes\fk\to\fk$ satisfies the axioms
of a Lie superalgebra. This is a direct consequence of the following
observations:
\begin{enumerate}[label=(\roman*)]
\item $\fk_{\bar 0}$ is a Lie algebra: this is just the Jacobi
  identity of the Lie bracket of vector fields;
\item $\fk_{\bar 0}$ acts on $\fk_{\bar 1}$, by property (ii) of the
  spinorial Lie derivative;
\item the Dirac current is a symmetric $\fk_{\bar 0}$-equivariant map,
  by property (v) of the spinorial Lie derivative;
\item for any $\varepsilon\in\fk_{\bar 1}$, with associated Dirac
  current $K=\kappa(\varepsilon,\varepsilon)$, we have from the
  definition of Killing spinor and \eqref{eq:nabla1} that
  \begin{equation*}
    \begin{split}
      \mathscr L_{K}\varepsilon&=\nabla_{K}\varepsilon+\sigma(A_K)\varepsilon\\
      &=-(\varphi\wedge  K)\vol\cdot\varepsilon + 2 g(\varphi,K) \vol
      \cdot \varepsilon + \imath_{\varphi} \three \cdot \varepsilon\\
      &=g(\varphi,K)\vol\cdot\varepsilon + \imath_{\varphi}\three\cdot\varepsilon\\
      &=-\varphi\cdot\three\cdot\varepsilon\\
      &=0~,
    \end{split}
  \end{equation*}
  where the last equality holds by Proposition~\ref{prop:clifford}.
  This is equivalent to the component of the Jacobi identity for $\fk$
  with three odd elements.
\end{enumerate}
The proof is thus completed.
\end{proof}

\subsection{The Killing superalgebra is a filtered deformation}
\label{sec:ksa-as-filtered}

We now show that the Killing superalgebra
$\fk = \fk_{\bar 0} \oplus \fk_{\bar 1}$ is a filtered deformation of
a $\ZZ$-graded subalgebra of the Poincaré superalgebra $\fp$. To this
aim, it is convenient to denote the triple $(a,b,\varphi)$
collectively by $\Phi$ and to abbreviate the Killing spinor equation
as $\nabla_Z\varepsilon=\beta_Z^{\Phi}\varepsilon$, where $\beta^\Phi$
is the $\End(S(M))$-valued one-form defined by
\begin{equation}
  \label{eq:end-valued-one-form}
  \beta_Z^\Phi \varepsilon = Z \cdot (a + b \vol) \cdot \varepsilon - (\varphi \wedge
  Z) \cdot \vol \cdot \varepsilon + 2 g(\varphi,Z) \vol \cdot \varepsilon\;,
\end{equation}
for all $Z \in \fX(M)$ and $\varepsilon\in\Gamma(S(M))$. The notation
is chosen to make contact with that of
Proposition~\ref{prop:killspin}. The reason for the superscript $\Phi$
is to distinguish $\beta^\Phi$ from the more general component $\beta$
of the filtered Lie brackets in \eqref{eq:KSAasFD} below. For a
similar reason we also introduce the $\fso(TM)$-valued symmetric
bilinear tensor $\gamma^{\Phi}$ on $S(M)$ given by
\begin{equation*}
\gamma^\Phi(\varepsilon,\varepsilon)(Z)=-2k(\beta_Z^\Phi\varepsilon,\varepsilon)\;,
\end{equation*}
for all $Z\in\fX(M)$ and $\varepsilon\in\Gamma(S(M))$.

Let $\eE = \eE_{\bar 0} \oplus \eE_{\bar 1}$ be the super vector
bundle with
\begin{equation*}
  \eE_{\bar 0} = TM \oplus \fso(TM) \qquad\text{and}\qquad \eE_{\bar
    1} = S(M)
\end{equation*}
and (even) connection $\eD$ defined on $\eE_{\bar 0}$ by
\cite{KostantHol,Geroch}
\begin{equation}
  \label{eq:superconnection}
  \eD_Z
  \begin{pmatrix}
    \xi \\ A
  \end{pmatrix} =
  \begin{pmatrix}
    \nabla_Z\xi + A(Z)\\ \nabla_ZA - R(Z,\xi)
  \end{pmatrix}  
\end{equation}
and on $\eE_{\bar 1}$ by the connection $D$ in
\eqref{eq:gravitinoconnection}.  A section $(\xi,A)$ of $\eE_{\bar 0}$
is parallel if and only if $\xi$ is a Killing vector and $A = -
\nabla\xi$, whereas a section $\varepsilon$ of $\eE_{\bar 1}$ is
parallel if and only if it is a Killing spinor.
Therefore $\fk$ is a subspace of the parallel sections of $\eE$:
$\fk_{\bar 1}$ are precisely the parallel sections of $\eE_{\bar 1}$,
whereas $\fk_{\bar   0}$ are the parallel sections of $\eE_{\bar 0}$
which in addition leave invariant the scalars $a$ and $b$ and the
vector field $\varphi$.

Parallel sections $\zeta$ of a vector bundle with connection are uniquely
determined by their value $\left.\zeta\right|_o$ at any given point $o \in M$.  (We tacitly
assume that $M$ is connected.)  Let us introduce the following
notation
\begin{equation*}
  (V,\eta) = (T_oM, \left.g\right|_o) \qquad \fso(V) = \fso(T_oM) \qquad S =
  S_o(M)~.
\end{equation*}
Therefore $\fk$ determines a subspace of $\eE_o = V \oplus \fso(V) \oplus
S$, which is the underlying vector space of the Poincaré superalgebra
$\fp$.  We recall that $\fp$ is a $\ZZ$-graded Lie superalgebra with Lie brackets given in equation~\eqref{eq:poincare} and that the $\ZZ$ and $\ZZ_2$ gradings are
compatible.

Let $(\xi,A_\xi)$, with $A_\xi = - \nabla \xi$, and $(\zeta, A_\zeta)$
belong to $\fk_{\bar 0}$.  Their Lie bracket is given by
\begin{equation}
  \label{eq:fk0LieBrackets}
  [(\xi,A_\xi), (\zeta, A_\zeta)] = (A_\xi \zeta - A_\zeta \xi,
  [A_\xi,A_\zeta] + R(\xi,\zeta))~,
\end{equation}
where the bracket on the right-hand side is the commutator in
$\fso(TM)$.  We see that the Riemann curvature measures the failure of
$\fk_{\bar 0}$ to be a Lie subalgebra of the Poincaré algebra
$\fp_{\bar 0}$.  If now $\varepsilon \in \fk_{\bar 1}$, then the Lie
bracket with $(\xi,A_\xi)$ is given by
\begin{equation}\label{eq:k0k1LieBracket}
  [(\xi,A_\xi), \varepsilon] = \nabla_\xi \varepsilon + \sigma(A_\xi) \varepsilon
  = \beta^\Phi_\xi \varepsilon + \sigma(A_\xi) \varepsilon~,
\end{equation}
where $\sigma : \fso(TM) \to \End (S(M))$ is the spinor representation.
Finally, the Dirac current of a Killing spinor $\varepsilon \in
\fk_{\bar 1}$ is given by
\begin{equation*}
  [\varepsilon,\varepsilon] = (\kappa(\varepsilon,\varepsilon), A_{\kappa(\varepsilon,\varepsilon)})~,
\end{equation*}
where
\begin{equation*}
  A_{\kappa(\varepsilon,\varepsilon)}(Z) = - \nabla_Z
  \kappa(\varepsilon,\varepsilon) = -2 \kappa(\nabla_Z
  \varepsilon,\varepsilon) = - 2 \kappa(\beta^\Phi_Z \varepsilon,
  \varepsilon)~.
\end{equation*}

We now show that $\fk$ defines a graded subspace of $\fp = \eE_o$.
Define $\ev_o^{\bar 0}: \fk_{\bar 0} \to V$ to be evaluation at $o$
and projection onto $V = T_oM$.  More precisely,
\begin{equation*}
  \ev_o^{\bar 0}(\xi,A_\xi) = \left.\xi\right|_o~.
\end{equation*}
Similarly, let $\ev_o^{\bar 1}: \fk_{\bar 1} \to S$ be the evaluation
at $o$. We set $S' = \im \ev_o^{\bar 1}$ and
$V' = \im \ev_o^{\bar 0}$.  

Let $\fh = \ker \ev_o^{\bar 0}$.  These are the Killing vectors in
$\fk_{\bar 0}$ which take the form $(0,A) \in V \oplus \fso(V)$ at $o
\in M$.  Therefore $\fh$ defines a subspace of $\fso(V)$, but from
equation~\eqref{eq:fk0LieBrackets}, we see that it is also a Lie
subalgebra:
\begin{equation*}
  [(0,A),(0,B)] = (0,[A,B])~.
\end{equation*}
In addition, the conditions $\eL_\xi a = \eL_\xi b = \eL_\xi \varphi =
0$ that are satisfied by the Killing vectors $\xi \in \fk_{\bar 0}$,
when evaluated at $o \in M$, imply that if $(0,A) \in \fh$ then
\begin{equation*}
  A \in \fso(V) \cap \stab(\left.a\right|_o) \cap \stab(\left.b\right|_o) \cap \stab(\left.\varphi\right|_o)~,
\end{equation*}
and the Lie
bracket \eqref{eq:k0k1LieBracket} at $o\in M$ implies that
\begin{equation*}
  [(0,A),\varepsilon] = \sigma(A) \varepsilon~.
\end{equation*}
In particular, $\fh$ acts on $S$ by restricting the action of
$\fso(V)$, and this action preserves $S'$.

The Lie subalgebra $\fh < \fk_{\bar 0}$ defines a short exact sequence
\begin{equation}
  \label{eq:ses}
  \begin{CD}
    0 @>>> \fh @>>> \fk_{\bar 0} @>{\ev_o^{\bar 0}}>> V' @>>> 0~,
  \end{CD}
\end{equation}
which yields a vector space isomorphism $\fk_{\bar 0} \cong \fh
\oplus V'$, and therefore \emph{as graded vector spaces}, a (non-canonical)
isomorphism
\begin{equation*}
  \fk \cong \fh \oplus S' \oplus V' \subset \fso(V) \oplus S \oplus
  V \cong \fp~.
\end{equation*}

We now wish to express the Lie superalgebra structure on $\fk$ in terms of a Lie
bracket on the graded vector space $\fh \oplus S' \oplus V'$.  This
requires a choice of splitting of the short exact sequence
\eqref{eq:ses}.  Geometrically, this amounts to choosing for every
$v \in V'$ a Killing vector field $\xi \in \fk_{\bar 0}$ with
$\left.\xi\right|_o = v$.  Such a choice gives an embedding of $V'$
into $V\oplus \fso(V)$ as the graph of a linear map
$\Sigma: V' \to \fso(V)$; that is, by sending $v \in V'$ to
$(v, \Sigma_v)$, where $\Sigma_v \in \fso(V)$ is the image of $v$
under $\Sigma$.  Any other choice of splitting would result in
$(v,\Sigma'_v)$ for some other linear map $\Sigma': V' \to \fso(V)$,
but where the difference $\Sigma - \Sigma' : V' \to \fh$.

The Lie bracket of $(0,A) \in \fh$ and $(v,\Sigma_v) \in \fk_{\bar 0}$
is given by
\begin{equation*}
  [(0,A),(v,\Sigma_v)] = (Av, [A,\Sigma_v]) = (Av, \Sigma_{Av}) +
  (0,[A,\Sigma_v]-\Sigma_{Av})~.
\end{equation*}
Similarly, if $\varepsilon \in \fk_{\bar 1}$, then
\begin{equation*}
  [(v,\Sigma_v), \varepsilon] = \beta^\Phi_v \varepsilon + \Sigma_v \varepsilon~,
\end{equation*}
whereas
\begin{equation*}
  [\varepsilon,\varepsilon] =
  (\kappa(\varepsilon,\varepsilon),A_{\kappa(\varepsilon,\varepsilon)})
  = (\kappa(\varepsilon,\varepsilon),
  \Sigma_{\kappa(\varepsilon,\varepsilon)}) + (0,
  A_{\kappa(\varepsilon,\varepsilon)} - \Sigma_{\kappa(\varepsilon,\varepsilon)})~.
\end{equation*}
Finally, if $v,w \in V'$,
\begin{equation*}
  \begin{split}
    [(v,\Sigma_v), (w,\Sigma_w)] &= (\Sigma_v w - \Sigma_w v,
    [\Sigma_v,\Sigma_w] + R(v,w))\\
    &= (\Sigma_v w - \Sigma_w v, \Sigma_{\Sigma_v w - \Sigma_w v}) + (0,
    [\Sigma_v,\Sigma_w] + R(v,w) - \Sigma_{\Sigma_v w - \Sigma_w v})
  \end{split}
\end{equation*}

This allows us to read off the Lie bracket on $\fh \oplus S' \oplus
V'$.  We will let $v,w \in V'$, $s \in S'$ and $A,B \in \fh$.  Then we have
\begin{equation}
  \label{eq:KSAasFD}
  \begin{split}
    [A,B] &= AB - BA\\
    [A,s] &= \sigma(A) s\\
    [A,v] &= Av + \underbrace{[A,\Sigma_v]-\Sigma_{Av}}_{\lambda(A,v)}\\
    [s,s] &= \kappa(s,s) + \underbrace{\gamma^\Phi(s,s) -
      \Sigma_{\kappa(s,s)}}_{\gamma(s,s)}\\
    [v,s] &= \underbrace{\beta^\Phi_v s + \Sigma_v s}_{\beta(v,s)}\\
    [v,w] &= \underbrace{\Sigma_v w - \Sigma_w v}_{\alpha(v,w)} +
    \underbrace{[\Sigma_v,\Sigma_w] + R(v,w) -
      \Sigma_{\alpha(v,w)}}_{\delta(v,w)}~,
  \end{split}
\end{equation}
which define maps $\lambda: \fh \otimes V' \to \fh$, $\gamma: \odot^2S' \to
\fh$, $\beta: V' \otimes S' \to S'$, $\alpha: \wedge^2 V' \to V'$ and
$\delta: \wedge^2 V' \to \fh$.

Notice that all the under-braced terms have positive
filtration degree: $\lambda$, $\alpha$, $\beta$ and $\gamma$ have degree
$2$, whereas $\delta$ has degree $4$.  If we set those maps to zero,
which is equivalent to passing to the associated graded superalgebra,
then we are left with the $\ZZ$-graded subalgebra $\fa < \fp$ given by
the Lie brackets
\begin{equation}
  \label{eq:undeformed}
    \begin{aligned}[m]
      [A,B] &= AB - BA\\
      [A,s] &= \sigma(A) s\\
      [A,v] &= Av
    \end{aligned}
    \qquad\qquad
    \begin{aligned}[m]
      [s,s] &= \kappa(s,s) \\
      [v,s] &= 0\\
      [v,w] &= 0~.
    \end{aligned}
\end{equation}
Moreover, it follows from
Lemma~\ref{lem:homogeneity} in the Appendix that if
$\dim S' > \frac12 \dim S = 2$, then $V' = V$.

Therefore we have proved the following

\begin{proposition}
  \label{prop:KSAisFD}
  The Killing superalgebra $\fk$ in equation~\eqref{eq:KSAasFD}  is a
  filtered deformation of the $\ZZ$-graded subalgebra $\fa < \fp$
  defined on $\fh \oplus S' \oplus V'$ by the Lie brackets in
  \eqref{eq:undeformed}. Moreover if $\dim S'>\frac12\dim S=2$ then the Lie algebra $\fk_{\bar 0}$ of infinitesimal automorphisms
  of $(M,g,a,b,\varphi)$ acts locally transitively around any point $o\in M$.
\end{proposition}

\section{Zero curvature equations}
\label{sec:zero-curvature-eqns}

In this section we calculate the curvature of the connection $D$ on
the spinor bundle and solve the zero curvature equations for the
metric $g$ and the fields $a,b,\varphi$.  We do this in two steps.  In
the first step we arrive at a first set of equations obtained by
setting the Clifford trace of the curvature to zero.  We perform this
first step for two reasons.  The first reason is by analogy with
eleven-dimensional supergravity, where the vanishing of 
the Clifford trace of the curvature is equivalent to the bosonic field
equations (and the Bianchi identity). The second reason is that this first set of
equations is easier to solve and already imposes strong constraints on
the geometric data which simplify the solution of the zero curvature
equations.  The second step is the solution of the zero curvature
equations, which will yield the maximally supersymmetric backgrounds.
The Killing superalgebras of these maximally supersymmetric
backgrounds should (and do) agree with the maximally supersymmetric
filtered deformations which we classify in Section \ref{sec:maxim-fil-def}.

With regard to the first reason for performing the first step, we must
stress that any relation in four dimensions between the equation
obtained by setting to zero the Clifford trace of the curvature and
the bosonic field equations of minimal off-shell supergravity remains
to be seen. If we were to identify (up to constants of
proportionality) the fields $a$, $b$ and $\varphi$ in the connection
$D$ in \eqref{eq:gravitinoconnection} with the bosonic fields in the
minimal off-shell gravity supermultiplet in four dimensions (as
described, say, in \cite[§16.2.3]{MR2920151}), and identify (up to
an overall constant of proportionality) $D \varepsilon$ with the
supersymmetry variation of the gravitino $\Psi$ in the gravity
supermultiplet, evaluated at $\Psi =0$, one finds that the purely
bosonic terms in the off-shell supergravity Lagrangian density must be
proportional to $R + 24( a^2 + b^2 + | \varphi |^2 )$, where $R$ is
the scalar curvature of $g$. The Einstein equations for this
supergravity Lagrangian are
$R_{\mu\nu} = -12 (a^2 + b^2) g_{\mu\nu} -24 \varphi_\mu \varphi_\nu$
which, after integrating out the auxiliary fields $a$, $b$ and
$\varphi$, imply that that $g$ must be Ricci-flat. As we will see, the
equations obtained by setting to zero the Clifford trace of the
curvature are similar but different.

\subsection{The curvature of the superconnection}
\label{sec:curv-superc}

Let us write the Killing spinor condition for $\varepsilon \in
\Gamma(S(M))$ as $\nabla_Z \varepsilon = \beta^\Phi_Z\varepsilon$ for all
vector fields $Z$, and where the $\End (S(M))$-valued one-form $\beta^\Phi$
was defined in equation~\eqref{eq:end-valued-one-form}.  In other
words, $D_Z = \nabla_Z - \beta^\Phi_Z$.  The curvature $R^D$ of $D$ is
defined by
\begin{equation*}
  \begin{split}
    R^D_{X,Y} &= D_{[X,Y]} - [D_X, D_Y]\\
    &= R_{X,Y} + (\nabla_X \beta^\Phi)_Y - (\nabla_Y \beta^\Phi)_X -
    [\beta^\Phi_X,\beta^\Phi_Y]~,
  \end{split}
\end{equation*}
where $R$ is the curvature $2$-form of $\nabla$ on the spinor bundle.
An explicit calculation shows that
\begin{equation}
  \label{eq:curvatureD}
  \begin{split}
    R^D_{X,Y} &= R_{X,Y} + X(a) Y + X(b) Y \cdot \vol - Y(a) X - Y(b) X
    \cdot \vol - (\nabla_X\varphi \wedge Y) \cdot \vol \\
    & \quad  {} + (\nabla_Y\varphi \wedge X) \cdot \vol + 2 g(\nabla_X
    \varphi, Y) \vol - 2 g(\nabla_Y \varphi, X) \vol
    + 2 (a^2 + b^2 - |\varphi|^2_g) X \wedge Y \\
    & \quad {}  + 4 a (\varphi \wedge X \wedge Y)
    \cdot \vol + 4 a g(\varphi,Y) X \cdot \vol - 4 a g(\varphi,X) Y
    \cdot \vol - 4 b \varphi \wedge X \wedge Y\\
    & \quad {} - 4 b g(\varphi, Y) X + 4
    b g(\varphi, X) Y + 2 g(\varphi, X) \varphi \wedge Y - 2 g(\varphi,
    Y) \varphi \wedge X~.
  \end{split}
\end{equation}

From this expression we will be able to read off a set of equations by
demanding that the Clifford trace of the curvature
$\Ric^D : TM \to \End (S(M))$, defined by
\begin{equation}
  \label{eq:DRic}
   \Ric^D(X) = \sum_\mu e^\mu \cdot R^D_{X,e_\mu}~,
\end{equation}
vanishes.  Here $e^\mu$ and $e_\mu$ are $g$-dual local frames of $TM$.  Another explicit
calculation shows that
\begin{equation}
  \label{eq:RicD}
  \begin{split}
    \Ric^D(X) &= \Ric(X) - 3 X(a) - 3 X(b) \vol - da^\sharp \wedge X - 
    (db^\sharp \wedge X) \cdot \vol + 6 ( a^2 + b^2 ) X \\
    & \quad {} - 4 |\varphi|^2_g X - 4 a (\varphi \wedge X)\cdot \vol
    + 4 b \varphi \wedge X + 12 a g(\varphi, X) \vol - 12 b g(\varphi,X) \\
    & \quad {} + 4 g(\varphi,X) \varphi + \left(\nabla_\mu \varphi_\nu X_\rho \Gamma^{\mu\nu\rho} -
    \nabla_\mu \varphi^\mu X - 2 g(\slashed\nabla\varphi, X)\right) \cdot \vol~,
  \end{split}
\end{equation}
where $\Ric$ stands for the Ricci operator and we have introduced the
shorthand $g(\slashed\nabla\varphi, X) =  \Gamma^\rho \nabla_\rho
\varphi_\mu X^\mu$.

\subsection{The vanishing of the Clifford trace of the curvature}
\label{sec:vanishing-Clifford-trace}

We now describe the equations arising by demanding that the Clifford
trace of the curvature of the spinor connection $D$ vanishes; in other
words, that for all vector fields $X$, $\Ric^D(X)=0$.  This is a
system of equations with values in $\End (S(M))$, which is isomorphic as
a vector bundle to $\bigoplus_{p=0}^4 \wedge^p TM$.  This means that
the components of these equations in each summand have to be satisfied
separately.  The $p=1$ component relates the Ricci tensor to the data
$(a,b,\varphi)$, whereas the $p\neq 1$ components constrain
$(a,b,\varphi)$.  We start with these first.

\subsubsection{The \texorpdfstring{$p=0$}{p=0} component}
\label{sec:p=0-component}

The $p=0$ component of the equation $\Ric^D(X) = 0$ is given by
\begin{equation*}
  -3 X(a) - 12 b g(\varphi, X) = 0~,
\end{equation*}
which, after abstracting $X$, is equivalent to
\begin{equation}
  \label{eq:0component}
  da^\sharp = -4 b \varphi~.
\end{equation}

\subsubsection{The \texorpdfstring{$p=4$}{p=4} component}
\label{sec:p=4-component}

The $p=4$ component of $\Ric^D(X)=0$ is given by
\begin{equation*}
  -3 X(b) \vol + 12 a g(\varphi,X) \vol = 0
\end{equation*}
which is equivalent to
\begin{equation}
  \label{eq:4component}
  db^\sharp = 4 a \varphi~.
\end{equation}

\subsubsection{The \texorpdfstring{$p=2$}{p=2} component}
\label{sec:p=2-component}

The $p=2$ component of $\Ric^D(X)=0$ is given by
\begin{equation*}
  -da^\sharp \wedge X - (db^\sharp \wedge X) \cdot \vol - 4 a (\varphi
  \wedge X) \cdot \vol + 4 b \varphi \wedge X = 0~,
\end{equation*}
which using equations~\eqref{eq:0component} and \eqref{eq:4component}
becomes
\begin{equation*}
  (\varphi \wedge X) \cdot (b + a \vol) = 0~.
\end{equation*}
Multiplying by $b-a\vol$ and since this has to be true for all $X$, we
arrive at
\begin{equation}
  \label{eq:2component}
  (a^2 + b^2) \varphi = 0~.
\end{equation}
It follows from this equation that there are three branches of
solutions:
\begin{enumerate}[label=(\Roman*)]
\item $a=b=\varphi =0$,
\item $a^2+b^2>0$ and $\varphi = 0$, in which case $a$ and $b$ are
  constant by equations~\eqref{eq:0component} and \eqref{eq:4component}, and
\item $a=b=0$ and $\varphi\neq 0$.
\end{enumerate}

\subsubsection{The \texorpdfstring{$p=3$}{p=3} component}
\label{sec:p=3-component}

The $p=3$ component of $\Ric^D(X) = 0$ is given by 
\begin{equation*}
  -(\nabla_\mu \varphi^\mu X + 2 g(\slashed\nabla \varphi, X)) \cdot \vol
  = 0~,
\end{equation*}
which, abstracting $X$, can be written as
\begin{equation}
  \label{eq:pre3component}
  \nabla_\mu \varphi^\mu \Gamma^\nu + 2 \nabla_\mu \varphi^\nu \Gamma^\mu = 0~.
\end{equation}
Multiplying with $\Gamma_\nu$ on both left and right we arrive at the pair
of equations:
\begin{equation*}
  \begin{split}
    -4 \nabla_\mu \varphi^\mu + 2 \nabla_\mu \varphi_\nu \Gamma^\mu \Gamma^\nu &= 0\\
    -4 \nabla_\mu \varphi^\mu + 2 \nabla_\mu \varphi_\nu \Gamma^\nu \Gamma^\mu &= 0~.
  \end{split}
\end{equation*}
Adding the two equations, and using the Clifford relations,
\begin{equation*}
  -8 \nabla_\mu \varphi^\mu - 4 \nabla_\mu \varphi^\mu = 0 \implies \nabla_\mu
  \varphi^\mu = 0~.
\end{equation*}
Plugging this back into equation~\eqref{eq:pre3component}, we arrive at
\begin{equation*}
  \nabla_\mu \varphi^\nu \Gamma^\mu = 0~,
\end{equation*}
which says
that $\varphi$ is parallel:
\begin{equation}
  \label{eq:3component}
  \nabla \varphi = 0~.
\end{equation}

\subsubsection{The \texorpdfstring{$p=1$}{p=1} component}
\label{sec:p=1-component}

Finally we arrive at the $p=1$ component of $\Ric^D(X) = 0$:
\begin{equation*}
  \Ric(X) + 6 ( a^2 + b^2 ) X - 4 |\varphi|^2_g X + 4 g(\varphi,X)
  \varphi + \nabla_\mu \varphi_\nu X_\rho \Gamma^{\mu\nu\rho} \cdot \vol = 0~.
\end{equation*}
The last term vanishes because $\varphi$ is parallel, so that we are
left with
\begin{equation*}
  \Ric(X) + 6 ( a^2 + b^2 ) X - 4 |\varphi|^2_g X + 4 g(\varphi,X)
  \varphi = 0~.
\end{equation*}
We can abstract $X$ and leave it as an equation on the Ricci operator
itself:
\begin{equation}
  \label{eq:1component}
  \Ric = - 12(a^2 + b^2) \id + 8 |\varphi|^2_g \id - 8 \varphi \otimes
  \varphi^\flat~,
\end{equation}
which, in terms of the symmetric Ricci tensor, becomes
\begin{equation}
  \label{eq:Ricci}
  R_{\mu\nu} = - 12 (a^2+b^2) g_{\mu\nu} + 8 |\varphi|^2_g g_{\mu\nu} - 8 \varphi_\mu
  \varphi_\nu~.
\end{equation}

\subsection{The solutions}
\label{sec:solutions}

Let us analyse the type of solutions to these equations.  We have
seen that there are three branches of solutions stemming from the
$p=2$ component equation~\eqref{eq:2component}.

\begin{enumerate}[label=(\Roman*)]
\item $a=b=\varphi =0$.  In this case, the $p=1$ component equation
  simply says that $g$ is Ricci-flat.  In this background, Killing
  spinors are parallel and therefore the supersymmetric backgrounds
  are the Ricci-flat manifolds whose holonomy is contained in the
  isotropy of a spinor.   Since the Dirac current of a parallel spinor
  is null and parallel, these metrics are Ricci-flat Brinkmann
  metrics.  See, e.g., \cite[§3.2.3]{JMWaves} for a discussion of
  these geometries.

\item $a^2+b^2\neq 0$ and $\varphi = 0$.  Putting $\varphi = 0$, we
  see from equations~\eqref{eq:0component} and \eqref{eq:4component}
  that $da = db = 0$, so they are constant  and the Ricci tensor is
  given by
  \begin{equation*}
    R_{\mu\nu} = - 12 (a^2+b^2) g_{\mu\nu}~,
  \end{equation*}
  so that $g$ is Einstein with negative cosmological constant. The
  Killing spinors are (up to an R-symmetry  which allows us to set
  $b=0$, say) geometric Killing spinors.  Such geometries are reviewed
  in \cite[§§6-7]{MR1822354} and discussed in \cite{MR2008949}.

\item $a=b=0$ and $\varphi\neq 0$.  Then $\varphi$ is a parallel
  vector field and the Ricci tensor, given by
  \begin{equation}\label{eq:RicciFluid}
    R_{\mu\nu} = -8 \left( \varphi_\mu \varphi_\nu - |\varphi|^2_g
      g_{\mu\nu} \right)~,
  \end{equation}
  is also parallel.  This is a kind of fluid solution.  Ricci-parallel
  geometries have been studied in \cite{MR1856414}.  The determining
  factor is the algebraic type of the Ricci endomorphism.  In this
  case, this depends on the causal type of $\varphi$, which is
  constant because $\varphi$ is parallel.  If $\varphi$ is timelike or
  spacelike, so that (in our mostly minus conventions) $|\varphi|_g^2$
  is positive or negative, respectively, then the Ricci endomorphism
  is diagonalisable and the geometry decomposes (up to coverings) into a
  product $M=\mathbb R\times N$ of a line and a three-dimensional
  Einstein space $N$, hence a space form. Moreover, upon identifying
  the spin bundle of $M$ with (an appropriate number of copies of) the
  spin bundle of $N$, it is not difficult to see that Killing spinors
  in these backgrounds correspond to geometric Killing spinors on $N$
  (up to an R-symmetry).  If $\varphi$ is null, then the Ricci
  endomorphism is two-step nilpotent and the geometry is
  Ricci-null. The subbundle of $TM$ of orthogonal vectors to $\varphi$
  is also in this case integrable in the sense of Frobenius but the
  above simple interpretation of Killing spinors is missing since the
  associated integrable submanifolds $N$ have a degenerate induced
  metric.
  \end{enumerate}

\subsection{Maximally supersymmetric backgrounds}
\label{sec:maxim-supersymm-back}

Maximally supersymmetric backgrounds are those for which the spinor
connection $D$ is flat. The zero curvature condition $R^D_{X,Y}=0$ for
all vector fields $X,Y$ becomes a system of equations with values in
$\End (S(M))$ and therefore, just as for the vanishing of the Clifford
trace of the curvature, the different components of the curvature must
vanish separately. We can reuse our calculations above, since if $D$
is flat, the Clifford trace of the curvature certainly vanishes. This
means that we can consider the three branches described above. We will
meet the geometries we are about to discuss again in the next section,
where we classify the maximally supersymmetric filtered
subdeformations of the Poincaré superalgebra.

\subsubsection{Maximally supersymmetric backgrounds with \texorpdfstring{$a=b=\varphi  = 0$}{a=b=phi=0}}
\label{sec:maxim-supersymm-back-1}

If $a=b=\varphi = 0$, the connection $D$ agrees with the Levi--Civita
spin connection and hence $D$-flatness means flatness and every such
background is locally isometric to Minkowski spacetime.

\subsubsection{Maximally supersymmetric backgrounds with \texorpdfstring{$\varphi = 0$}{phi=0} and \texorpdfstring{$a^2+b^2>0$}{a2 + b2 > 0}}
\label{sec:maxim-supersymm-back-2}

If $\varphi = 0$, then $a,b$ are constant and not both zero and hence
the $D$-flatness condition is
\begin{equation*}
  R_{X,Y} = - 2 (a^2 + b^2) X \wedge Y~,
\end{equation*}
as an equation in $\End(S(M))$.   This is equivalent to
\begin{equation*}
  R_{\mu\nu\rho\sigma} = 4 (a^2 + b^2) (g_{\mu\rho}g_{\nu\sigma} -
  g_{\mu\sigma} g_{\nu\rho})~,
\end{equation*}
which says that $g$ is locally isometric to $\AdS_4$.

\subsubsection{Maximally supersymmetric backgrounds with \texorpdfstring{$a=b=0$}{a=b=0} and
  \texorpdfstring{$\varphi\neq 0$}{phi not = 0}}
\label{sec:maxim-supersymm-back-3}

If $a=b=0$, and using that $\varphi$ is parallel, the $D$-flatness
condition is
\begin{equation*}
    R_{X,Y} = 2 |\varphi|^2_g X \wedge Y - 2 g(\varphi, X) \varphi
    \wedge Y + 2 g(\varphi, Y) \varphi \wedge X~,
\end{equation*}
again as an equation in $\End(S(M))$.  The corresponding Riemann
tensor is given by
\begin{equation}
  \label{eq:riemannLieGroup}
  R_{\mu\nu\rho\sigma} = - 4 |\varphi|^2_g (g_{\mu\rho} g_{\nu\sigma} - g_{\mu\sigma} g_{\nu\rho}) - 4
  \varphi_\nu\varphi_\rho g_{\mu\sigma} + 4 \varphi_\nu \varphi_\sigma g_{\mu\rho} + 4
  \varphi_\mu \varphi_\rho g_{\nu\sigma} - 4 \varphi_\mu \varphi_\sigma g_{\nu\rho}~.
\end{equation}
Since $\varphi$ and $g$ are parallel, so is the Riemann tensor and
hence this corresponds to a locally symmetric space.  Furthermore, it
is conformally flat.  Indeed, in four dimensions, the Weyl tensor is
given in terms of the Riemann tensor, the Ricci tensor $R_{\mu\nu} =
g^{\rho\sigma} R_{\mu\rho\sigma\nu}$ and the Ricci scalar $R =
g^{\mu\nu}R_{\mu\nu}$ by
\begin{equation*}
    W_{\mu\nu\rho\sigma} = R_{\mu\nu\rho\sigma} + \tfrac12 \left( g_{\mu\rho} R_{\nu\sigma} - g_{\mu\sigma}
      R_{\nu\rho} - g_{\nu\rho} R_{\mu\sigma} + g_{\nu\sigma} R_{\mu\rho} \right) - \tfrac16 R
    \left(g_{\mu\rho} g_{\nu\sigma} - g_{\mu\sigma} g_{\nu\rho}\right)~.
\end{equation*}
Inserting the above expression for $R_{\mu\nu\rho\sigma}$ into the Weyl tensor we
see that it vanishes, so that the geometry is conformally flat.  The
corresponding Ricci tensor is given by equation~\eqref{eq:RicciFluid}
and the Ricci scalar is $R = 24 |\varphi|^2_g$.

This geometry corresponds to a Lorentzian Lie group with a
bi-invariant metric. Indeed, the equation \eqref{eq:riemannLieGroup}
satisfied by the Riemann tensor is equivalent to the vanishing of the
curvature of a metric connection with parallel totally skewsymmetric
torsion proportional to the Hodge dual of $\varphi$. As shown, for
instance, in \cite{CFOSchiral,JMFPara}, the existence of a flat metric
connection with closed skewsymmetric torsion is equivalent to the
manifold being locally isometric to a Lie group with a bi-invariant
metric.

Since $\varphi$ is parallel, its $g$-norm is constant and in a
Lorentzian manifold this can be of three types:
\begin{enumerate}
\item $|\varphi|^2_g > 0$.  This is timelike in our conventions.  The
  background is locally isometric to $\RR \times \Sph^3$, where we
  identify the round $\Sph^3$ with the Lie group $\SU(2)$ with its
  bi-invariant metric.

\item $|\varphi|^2_g < 0$.  This is spacelike and hence the background
  is locally isometric to $\AdS_3 \times \RR$, where we identify
  $\AdS_3$ with $\SL(2,\RR)$ with its bi-invariant metric.

\item $|\varphi|^2_g = 0$.  This is the null case and hence the
  background is locally isometric to the Nappi--Witten group
  \cite{NW} with its bi-invariant metric.
\end{enumerate}

\section{Maximally supersymmetric filtered deformations}
\label{sec:maxim-fil-def}

We now resume the analysis of filtered subdeformations of the Poincaré
superalgebra by classifying the filtered deformations with maximal
odd dimension. We will show that they correspond precisely to the
Killing superalgebras of the maximally supersymmetric backgrounds
classified in Section~\ref{sec:maxim-supersymm-back}.

More precisely, let $\fa=\fa_{-2}\oplus\fa_{-1}\oplus\fa_{0}$ be a
$\mathbb Z$-graded subalgebra of the Poincaré superalgebra
$\fp=V\oplus S\oplus\fso(V)$ with $\fa_{-1} = S$. By
Lemma~\ref{lem:homogeneity}, we also have that $\fa_{-2} = V$, so that
$\fa$ differs from $\fp$ only in zero degree, where $\fa_{0}=\fh$ is a
subalgebra of $\fso(V)$. The aim of this section is to classify, for
any possible given $\fh$, the filtered deformations $\fg$ of $\fa$. We
will see that they are essentially governed by the $\fh$-invariant
elements $H^{2,2}(\fa_{-},\fa)^\fh$ of the Spencer group
$H^{2,2}(\fa_{-},\fa)$ of $\fa$, where
$\fa_{-}=\fa_{-2}\oplus\fa_{-1}$ is the negatively graded part of
$\fa$.

In Section~\ref{sec:preliminariesA} we set up the calculation of
$H^{2,2}(\fa_{-},\fa)$, which will be described in
Section~\ref{sec:cohomologyA}.  This result will then be used in
Section~\ref{sec:integr-deform} to classify the filtered deformations.
The results are summarised in Theorem~\ref{thm:final} in
Section~\ref{sec:summary}.

\subsection{Preliminaries}
\label{sec:preliminariesA}

Here we set up the calculation of the Spencer cohomology
$H^{2,2}(\fa_{-},\fa)$. We introduce the Spencer complex of $\fa$ in
complete analogy to the Spencer complex of $\fp$ (cf.
Section~\ref{sec:spencer-cohomology}): one has simply to replace
$\fso(V)$ with $\fh$ in the definitions. For instance any element
$\zeta\in C^{2,2}(\fa_-,\fa)$ can be uniquely written as the sum
$\zeta=\alpha+\beta+\gamma$, where
\begin{equation}
  \label{eq:componentsA}
  \alpha \in \Hom(\wedge^2V,V)~,\qquad
  \beta \in\Hom(V \otimes S, S)\qquad\text{and}\qquad
  \gamma \in \Hom(\odot^2S, \fh)
\end{equation}
and the Lie brackets of a general filtered deformations of $\fa$ take the form
\begin{equation}
  \begin{aligned}[m]
    [A,B]&=AB-BA\\
    [A,s]&=\sigma(A)s\\
    [A,v]&=Av+\lambda(A,v)
  \end{aligned}
  \qquad\qquad
  \begin{aligned}[m]
    [s,s]&=\kappa(s,s)+\gamma(s,s)\\
    [v,s]&=\beta(v,s)\\
    [v,w]&=\alpha(v,w)+\delta(v,w)\;,
  \end{aligned}
\end{equation}
for some maps $\lambda:\fh\otimes V\to\fh$ and $\delta:\wedge^2
V\to\fh$, where $A,B\in \fh$, $s\in S$, $v,w\in V$.

We recall that a transitive and fundamental $\ZZ$-graded Lie
superalgebra $\fa=\bigoplus\fa_{p}$ with negatively graded part
$\fa_{-}=\bigoplus_{p<0}\fa_{p}$ is called a \emph{full prolongation
  of degree k} if $H^{d,1}(\fa_-,\fa)=0$ for all $d\geq k$.

\begin{lemma}
  \label{lem:fullprol}
  Let $\fa=\fa_{-2}\oplus\fa_{-1}\oplus\fa_0$ be a $\mathbb Z$-graded
  subalgebra of the Poincaré superalgebra which differs only in zero
  degree.  Then $\fa$ is fundamental, transitive and
  $H^{d,2}(\fa_-,\fa)=0$ for all even $d>2$. Furthermore it is a full
  prolongation of degree $k=2$.
\end{lemma}

\begin{proof}
  We only show  the last claim, the others follow as in the proof of Lemma~\ref{lem:p4}.
  Any $\zeta\in C^{2,1}(\fa_-,\fa)$ satisfies $\zeta(V)\subset\fh$, $\zeta(S)\subset \fa_{1}=0$ and
  \begin{align*}
    &\partial\zeta(s_1,s_2)=-\zeta(k(s_1,s_2))\\
    &\partial\zeta(s_1,v_1)=-\sigma( \zeta(v_1) )s_1\\
    &\partial\zeta(v_1,v_2)=\zeta(v_1)v_2-\zeta(v_2)v_1
  \end{align*}
  for all $s_1,s_2\in S$, $v_1,v_2\in V$.
  The first equation directly implies that $\zeta=0$ is the only
  cocycle and hence $H^{2,1}(\fa_-,\fa)=0$. If $d>2$ then
  $C^{d,1}(\fa_-,\fa)=0$ for degree reasons. 
\end{proof}

\begin{remark}
  One can actually prove that $\fa$ is a full prolongation of degree
  $k=1$, based on the non-trivial fact that the so-called ``maximal
  prolongation'' $\fg^\infty$ of $\fa_-=V\oplus S$ is a simple Lie
  superalgebra of type $\fsl(1|4)$ with a special $\ZZ$
  grading of the form
  $\fg^\infty=\fg^\infty_{-2}\oplus\cdots\oplus\fg^\infty_{2}$, cf.
  \cite{MR3255456}; but the simpler result of Lemma~\ref{lem:fullprol}
  suffices for our purposes.
\end{remark}

To state the main first intermediate result on filtered deformations
$\fg$ of $\fa$ we recall that the Lie brackets of $\fg$ have
components of nonzero degree: the sum $\mu:\fa\otimes\fa\to\fa$ of all
components of degree $2$ and the unique component
$\delta:\wedge^2 V\to \fh$ of degree $4$.

\begin{proposition}
  \label{thm:filt1}
  Let $\fa=V\oplus S\oplus \fh$ be a $\mathbb Z$-graded subalgebra of
  the Poincaré superalgebra $\fp=V\oplus S\oplus \fso(V)$ which
  differs only in zero degree. If $\fg$ is a filtered deformation of
  $\fa$ then:
  \begin{enumerate}
  \item $\mu|_{\fa_-\otimes\fa_-}$ is a cocycle in $C^{2,2}(\fa_-,\fa)$ and
    its cohomology class $$[\mu|_{\fa_-\otimes\fa_-}]\in H^{2,2}(\fa_-,\fa)$$
    is $\fh$-invariant (that is, the cocycle $\mu|_{\fa_-\otimes\fa_-}$ is
    $\fh$-invariant up to coboundaries); and
  \item if $\fg'$ is another filtered deformation of $\fa$ such that
    $[\mu'|_{\fa_-\otimes\fa_-}]= [\mu|_{\fa_-\otimes\fa_-}]$ then $\fg'$ is
    isomorphic to $\fg$ as a filtered Lie superalgebra.
  \end{enumerate}
\end{proposition}

\begin{proof}
  The first claim follows directly from Proposition~2.2 of
  \cite{MR1688484}. Let now $\fg$ and $\fg'$ be filtered deformations
  of $\fa$ such that
  $[\mu|_{\fa_-\otimes\fa_-}]=[\mu'|_{\fa_-\otimes\fa_-}]$. Then
  $(\mu-\mu')|_{\fa_-\otimes\fa_-}$ is a Spencer coboundary and we may
  first assume without any loss of generality that
  $\mu|_{\fa_-\otimes\fa_-}=\mu'|_{\fa_-\otimes\fa_-}$ by Proposition~2.3 of
  \cite{MR1688484}. Moreover, since $\fa$ is a fundamental and transitive full prolongation of
  degree $k=2$ by Lemma~\ref{lem:fullprol},
  Proposition~2.6 of \cite{MR1688484} applies and we may also assume
  $\mu=\mu'$ without any loss of generality. In other words we just
  showed that $\fg'$ is isomorphic as a filtered Lie superalgebra to
  another filtered Lie superalgebra $\fg''$ which satisfies $\mu''=\mu$.
  
  Now, given any two filtered deformations $\fg$ and $\fg'$ of $\fa$ with
  $\mu=\mu'$ it is easy to see that
  $\delta-\delta'=(\delta-\delta')|_{\fa_-\otimes\fa_-}$ is a Spencer
  cocycle (use e.g., \cite[equation 2.6]{MR1688484}). However
  $H^{4,2}(\fa_-,\fa)=\ker\partial|_{C^{4,2}(\fa_-,\fa)}=0$
  by Lemma~\ref{lem:fullprol} and hence $\delta=\delta'$. This 
  proves that any two filtered deformations $\fg$ and $\fg'$ of $\fa$ with
  $[\mu'|_{\fa_-\otimes\fa_-}]= [\mu|_{\fa_-\otimes\fa_-}]$ are isomorphic.
\end{proof}

In other words, filtered deformations are determined by the space
$H^{2,2}(\fa_-,\fa)^\fh$ of $\fh$-invariant elements in
$H^{2,2}(\fa_-,\fa)$.  In particular the components of non-zero
filtration degree $\lambda=\mu|_{\fh\otimes V}:\fh\otimes V\to \fh$
and $\delta:\wedge^2V\to \fh$ are completely determined by the class
$[\mu|_{\fa_-\otimes\fa_-}]\in H^{2,2}(\fa_-,\fa)^\fh$, up to
isomorphisms of filtered Lie superalgebras.

We will now describe $H^{2,2}(\fa_-,\fa)$. We recall that this group
has already been determined in Proposition~\ref{prop:killspin} when
$\fa=\fp$ is the Poincaré superalgebra. Therein we also described the
kernel $\mathscr H^{2,2}$ of the Spencer operator acting on
$\Hom(V\otimes S,S)\oplus\Hom(\odot^2 S, \fso(V))$: it consists of the
maps
$\beta+\gamma\in \Hom(V\otimes S,S)\oplus\Hom(\odot^2 S, \fso(V))$
which are of the form given by Proposition~\ref{prop:killspin}. To
avoid confusion with the general components \eqref{eq:componentsA} we
will denote these maps by $\beta^{\Phi}+\gamma^{\Phi}$ from now on,
that is we set $\Phi=(a,b,\varphi)\in2\mathbb R\oplus V$ and
\begin{equation*}
  \begin{split}
    \beta^{\Phi}(v,s)&=v\cdot(a+b\vol)\cdot
    s-\frac{1}{2}(v\cdot\varphi+3\varphi\cdot v)\cdot\vol\cdot s\;,\\
    \gamma^{\Phi}(s,s)v &=-2\kappa(s,\beta(v,s))\;,
  \end{split}
\end{equation*}
for all $v\in V$ and $s\in S$, according to Proposition~\ref{prop:killspin}. 
In addition we set
  \begin{align*}
    \gamma^\varphi(s,s)v&=2\kappa(s,(\varphi\wedge v)\cdot\vol\cdot s)\;,\\
    \gamma^{(a,b)}(s,s)v&=-2a\kappa(s,v\cdot s)
                          -2b\kappa(s,v\cdot\vol\cdot s)\;,
  \end{align*}
  for all $v\in V$, $s\in S$.

We will also determine the $\fh$-invariant classes in
$H^{2,2}(\fa_-,\fa)$, the Lie subalgebras $\fh\subset\fso(V)$ for
which $H^{2,2}(\fa_-,\fa)^{\fh} \neq 0$, hence the graded
subalgebras $\fa=V\oplus S\oplus\fh$ of $\fp$ admitting nontrivial
filtered deformations.  The condition
$H^{2,2}(\fa_-,\fa)^{\fh} \neq 0$ 
has strong consequences and, as we will now see,
gives rise to a dichotomy: either $\varphi=0$ and $a^2+b^2\neq 0$ or
$\varphi\neq 0$ and $a=b=0$.

\subsection{The cohomology group \texorpdfstring{$H^{2,2}(\fa_-,\fa)$}{H22(a-,a)}}
\label{sec:cohomologyA}

We start with the following

\begin{proposition}
  \label{prop:cohgroups}
  Let $\fa=V\oplus S\oplus \fh$ be a $\mathbb Z$-graded subalgebra of
  the Poincaré superalgebra $\fp=V\oplus S\oplus
  \fso(V)$ which differs only in zero degree. Then
  \begin{equation*}
    H^{2,2}(\fa_-,\fa) = \frac{\left\{\beta^{\Phi} + \gamma^{\Phi}
        + \partial\widetilde\psi\,\middle |\,\Phi\in 2\mathbb R\oplus V
        ,\,\widetilde\psi: V \to \fso(V)\;\text{s.t.}\;
        \gamma^{\Phi}(s,s)-\widetilde\psi(\kappa(s,s))\in
        \fh\right\}}{\left\{\partial\psi \, \middle | \, \psi: V \to \fh\right\}}
  \end{equation*}
  and
  \begin{itemize}
  \item[(i)] the cohomology class
    $[\beta^{\Phi} + \gamma^{\Phi} + \partial\widetilde\psi]$ is trivial if and only if $\Phi=0$;
  \item[(ii)] the condition 
    $\gamma^{\Phi}(s,s)-\widetilde\psi(\kappa(s,s))\in\fh$ 
    is satisfied for all $s\in S$ if and only if separately
    \begin{align}
      \label{eq:relphi}
      \gamma^\varphi(s,s)- \widetilde\psi(\kappa(s,s))&\in \fh\;,\\
      \label{eq:relab}
      \gamma^{(a,b)}(s,s)&\in \fh\;,
    \end{align}
    for all $s\in S$;
  \item[(iii)] if $[\beta^{\Phi} + \gamma^{\Phi} + \partial\widetilde\psi]$ is an $\fh$-invariant
    cohomology class then $\fh$ leaves
    $\varphi$ invariant, that is $\fh\subset \fh_\varphi$ where
    $\fh_\varphi = \fso(V) \cap \stab(\varphi)$ and $\stab(\varphi)$
    is the Lie algebra of the stabiliser of $\varphi$ in $\mathrm{GL}(V)$.  
  \end{itemize}
  In particular if $[\beta^{\Phi} + \gamma^{\Phi} + \partial\widetilde\psi] \in
  H^{2,2}(\fa_-,\fa)$ is a nontrivial and $\fh$-invariant
  cohomology class then exactly one of
  the following two cases occurs:
  \begin{itemize}
  \item[(1)] if $\varphi=0$ then $a^2+b^2\neq 0$,
    $\gamma^{\Phi}(s,s)=\gamma^{(a,b)}(s,s)\in \fh$
    for all $s\in S$ and the cohomology class $[\beta^{\Phi} +
    \gamma^{\Phi} +
    \partial\widetilde\psi]=[\beta^{\Phi}+\gamma^{\Phi}]$;
  \item[(2)] if $\varphi\neq 0$ then $a=b=0$ and 
    \begin{align}
      \label{eq:relphi2}
      \gamma^\varphi(s,s)&\in \fh_\varphi\;,\\
      \label{eq:relab2}
      \widetilde\psi(\kappa(s,s))&\in \fh_\varphi\;,
    \end{align}
    for all $s\in S$.
  \end{itemize}
\end{proposition}

\begin{proof}
  From Lemma~\ref{lem:iso} we know that given any $\alpha \in
  \Hom(\wedge^2V,V)$,
  there is a unique $\widetilde\psi\in\Hom(V,\fso(V))$ such that
  $\partial \widetilde\psi = \alpha + \widetilde\beta +
  \widetilde\gamma$, for some $\widetilde\beta \in \Hom(V \otimes S,
  S)$ and $\widetilde\gamma \in \Hom(\odot^2S,\fso(V))$. Any cochain
  $\alpha + \beta + \gamma\in C^{2,2}(\fa_-,\fa)$ may be therefore
  uniquely written as
  \begin{equation*}
    \alpha + \beta + \gamma = (\alpha + \beta + \gamma
    - \partial\widetilde\psi) + \partial\widetilde\psi = (\beta
    -\widetilde\beta) +  (\gamma - \widetilde\gamma) + \partial \widetilde\psi~,
  \end{equation*}
  where $\beta - \widetilde\beta \in \Hom(V\otimes S, S)$ and
  $\gamma-\widetilde\gamma \in \Hom(\odot^2S, \fso(V))$.  If $\alpha +
  \beta + \gamma$ is a cocycle, then so is $(\beta - \widetilde \beta) +
  (\gamma - \widetilde\gamma)$, so that by Proposition~\ref{prop:killspin},
  $\beta - \widetilde \beta = \beta^\Phi$ and $\gamma - \widetilde \gamma =
  \gamma^\Phi$ for some $\Phi \in 2\mathbb R\oplus V$ or,  in other words,
  \begin{equation}
    \label{eq:firstinclusion}
    \ker\partial\bigr|_{C^{2,2}(\fa_-,\fa)} \subset \mathscr H^{2,2}
    \oplus \partial\Hom(V,\fso(V))~\;.
  \end{equation}
  Conversely equation \eqref{eq:Spencer1} tells
  us that $\partial\widetilde\psi(s,s)=-\widetilde\psi(\kappa(s,s))$ 
  for all $s\in S$ so that an element $\beta^\Phi + \gamma^\Phi
  + \partial\widetilde\psi$ is in $C^{2,2}(\fa_-,\fa)$ if and only if
  \begin{equation}
    \label{eq:condPhi}
    \gamma^\Phi(s,s) - \widetilde\psi(\kappa(s,s))\in \fh\;,
  \end{equation}
  for all $s\in S$. This fact together with
  \eqref{eq:firstinclusion} yield immediately the claim on $H^{2,2}(\fa_-,\fa)$.

  If $\Phi=0$, then $\widetilde\psi(\kappa(s,s))\in\fh$ for all
  $s\in S$ and $\partial\widetilde\psi$ is in the image of
  $C^{2,1}(\fa_-,\fa)=\Hom(V,\fh)$, proving one implication of claim
  (i).  The other implication is trivial.
  
  We will now have a closer look at condition \eqref{eq:condPhi},
  using that $\odot^2 S=\wedge^1 V\oplus\wedge^2 V$. From (ii) of
  Proposition~\ref{prop:killspin} we have 
  \begin{align*}
    \gamma^{\Phi}(s,s)v&=-2\kappa(s,\beta^{\Phi}(v,s))\\
                       &=-2a\kappa(s,v\cdot s)
                         -2b\kappa(s,v\cdot\vol\cdot s)
                         +2\kappa(s,(\varphi\wedge v-2\eta(\varphi,v))\cdot\vol\cdot s)\\
                       &=-2a\kappa(s,v\cdot s)
                         -2b\kappa(s,v\cdot\vol\cdot s)
                         +2\kappa(s,(\varphi\wedge v)\cdot\vol\cdot s)\;,
  \end{align*}
  with the first two terms (resp. last term) in the RHS of the above
  equation acting on the component $\wedge^2 V$ (resp. $\wedge^1 V$)
  of $\odot^2 S$ but trivially on the other component $\wedge^1 V$
  (resp. $\wedge^2 V$). In particular condition \eqref{eq:condPhi}
  splits into \eqref{eq:relphi} and \eqref{eq:relab}, proving claim (ii).

  Let now $[\beta^{\Phi} + \gamma^{\Phi} + \partial\widetilde\psi] \in
  H^{2,2}(\fa_-,\fa)$ be an $\fh$-invariant class; i.e., for any
  $x\in\fh$ there is a $\psi\in\Hom(V,\fh)$ such that
  $x\cdot(\beta^\Phi+\gamma^\Phi+\partial\widetilde\psi)=\partial\psi$.
  In other words, in terms of the $\fso(V)$-equivariant projections
  \eqref{eq:projectors}, we have:
  \begin{align}
    \label{eq:a}
    x\cdot(\pi^{\alpha}(\partial\widetilde\psi))&=\pi^{\alpha}(\partial\psi)~,\\
    \label{eq:b}
    x\cdot(\beta^\Phi+\pi^\beta(\partial\widetilde\psi))&=\pi^{\beta}(\partial\psi)~,\\
    \label{eq:c}
    x\cdot(\gamma^\Phi+\pi^{\gamma}(\partial\widetilde\psi))&=\pi^{\gamma}(\partial\psi)~.
  \end{align}
  Equation~\eqref{eq:a} and the $\fso(V)$-equivariance of
  $\pi^\alpha$ and $\partial$ imply
  \begin{equation*}
    (\pi^{\alpha}\circ\partial)(\psi) =
    (\pi^{\alpha}\circ\partial) (x\cdot\widetilde\psi)
  \end{equation*}
  so that $x\cdot\widetilde\psi=\psi$, by Lemma~\ref{lem:iso}.
  Equation \eqref{eq:b} yields therefore
  \begin{equation*}
    \begin{split}
      \pi^{\beta}(\partial\psi) &= x\cdot \left(\beta^\Phi +
        \pi^{\beta} (\partial\widetilde\psi)\right)
      = x \cdot \beta^\Phi + x \cdot \pi^{\beta}(\partial\widetilde\psi) \\
      &= x\cdot\beta^\Phi+\pi^{\beta}(\partial(x\cdot
      \widetilde\psi))=x\cdot\beta^\Phi+\pi^{\beta}(\partial\psi)
    \end{split}
  \end{equation*}
  from which $\beta^{x\cdot \varphi}=x\cdot
  \beta^{\varphi}=x\cdot\beta^\Phi=0$. This proves claim (iii).
  
  We now prove the last claims. Let $[\beta^{\Phi} + \gamma^{\Phi} +
  \partial\widetilde\psi]$ be a nontrivial $\fh$-invariant class.  If
  $\varphi=0$ then $a^2+b^2\neq 0$ by (i) and $\partial\widetilde\psi$
  is in the image of $C^{2,1}(\fa_-,\fa)=\Hom(V,\fh)$ by
  \eqref{eq:relphi}. This fact together with \eqref{eq:relab}
  immediately gives case (1).
  
  If $\varphi\neq 0$ then (ii) and (iii) imply
  \begin{align}
    \label{eq:neq0}
    \gamma^\varphi(s,s)- \widetilde\psi(\kappa(s,s))&\in \fh_{\varphi}\;,\\
    \label{eq:neq0II}
    \gamma^{(a,b)}(s,s)&\in \fh_{\varphi}\;,
  \end{align}
  for all $s\in S$. We fix an orthonormal basis $\left\{e_\mu\right\}$
  of $V$, use the Einstein summation convention and note that equation
  \eqref{eq:neq0II} gives
  \begin{align*}
    0&=\gamma^{(a,b)}(s,s)\varphi\\
     &=2\varphi^\mu(\bar s\Gamma_\mu\Gamma_\nu(a+b\vol)s) e^\nu\\
     &=2a\varphi^\mu(\bar s\Gamma_\mu\Gamma_\nu s) e^\nu+
       2b\varphi^\mu(\bar s\Gamma_\mu\Gamma_\nu \vol s) e^\nu\\
     &=2a\varphi^\mu(\bar s\Gamma_{\mu\nu}s) e^\nu+
       2b\varphi^\mu(\bar s\Gamma_{\mu\nu} \vol s) e^\nu\\
     &=2\bar s(\varphi^\mu(a\Gamma_{\mu\nu}+b\Gamma_{\mu\nu} \vol)) s e^\nu\,,
  \end{align*}
  for all $s\in S$, hence
  $\varphi^\mu(a\Gamma_{\mu\nu}+b\Gamma_{\mu\nu} \vol)=0$ for every
  $0\leq j\leq 3$. Since $\varphi\neq 0$ this readily implies $a=b=0$.
  Similarly
  \begin{align*}
    \gamma^{\varphi}(s,s)\varphi&=-\varphi^\mu(\bar s\Gamma_\mu(\Gamma_\nu\varphi+3\varphi\Gamma_\nu)\vol s) e^\nu\\
                                &=-2\varphi^\mu\varphi^\rho(\bar s\Gamma_\mu\Gamma_{\rho\nu}\vol s) e^\nu\\
                                &=-2\varphi^\mu\varphi^\rho(\bar s\Gamma_{\mu\rho\nu}\vol s) e^\nu\\
                                &=0\;,
  \end{align*}
  so that  $\gamma^{\varphi}(s,s)\in\fh_{\varphi}$ for all $s\in S$
  automatically and, from equation \eqref{eq:neq0}, we infer that
  $\widetilde\psi(\kappa(s,s))\in\fh_{\varphi}$ for all $s\in S$ too.
  This is case (2).
\end{proof}

By the results of Proposition~\ref{thm:filt1} and Proposition~\ref{prop:cohgroups}, we need only
to consider the filtered deformations associated to $\fh$-invariant cohomology
classes in $H^{2,2}(\fa_-,\fa)$ with $\Phi\neq 0$. Indeed if
$\Phi=0$ then $[\mu|_{\fa_-\otimes\fa_-}]=0$ and the associated
filtered Lie superalgebras are just the $\mathbb Z$-graded subalgebras of the Poincaré superalgebra. 

We now investigate separately the cohomology classes in family (1) and (2) of Proposition~\ref{prop:cohgroups}.
\begin{lemma}
\label{lem:case1}
Let $[\beta^{\Phi} + \gamma^{\Phi}] \in
  H^{2,2}(\fa_-,\fa)$ be a nontrivial and $\fh$-invariant
  cohomology class with $\varphi=0$. Then $\fh=\operatorname{Im}(\gamma^{\Phi})=\fso(V)$.
\end{lemma}
\begin{proof}
First of all, as $a^2+b^2\neq 0$ by Proposition~\ref{prop:cohgroups}, we have that right multiplication by $a+b\vol$ in $\Cl(V)$ 
is a linear isomorphism. In particular it restricts to a linear isomorphism of $\wedge^2 V\subset\Cl(V)$. 
On the other hand, from Proposition~\ref{prop:killspin}:
\begin{align*}
\eta(w,\gamma^{\Phi}(s,s)v)&=-2\eta(w,\kappa(s,v\cdot(a+b\vol)\cdot s))\\
&=-2\left< s, w\cdot v\cdot(a+b\vol)\cdot s\right>\\
&=-2\left< s, w\wedge v\cdot (a+b\vol)\cdot s\right>\;,
\end{align*}
for all $w\wedge v\in\wedge^2 V\subset \Cl(V)$ and $s\in S$. Since $\gamma^{\Phi}(s,s)\in\fh$ for all $s\in S$ from (1) of Proposition~\ref{prop:cohgroups} and $\odot^2 S=\wedge^1 V\oplus\wedge^2 V$, the claim follows.
\end{proof}
To proceed further, we need to consider the case where
$\varphi\neq 0$, $a=b=0$. It is however sufficient to consider
$\varphi$ up to the action of
$\mathrm{CSO}(V) = \RR^\times \times \SO(V)$. To see it, we note that
the group $\mathrm{CSpin(V)}$ with Lie algebra $\mathfrak{co}(V)$ is a
double-cover of $\mathrm{CSO}(V)$ and it naturally acts on the
Poincaré superalgebra $\fp=V\oplus S\oplus\fso(V)$ by $0$-degree Lie
superalgebra automorphisms ($t\id \in\mathrm{CSpin(V)}$
acts with eigenvalues $0,e^{-t}$ and $e^{-2t}$ on, respectively,
$\fso(V)$, $S$ and $V$). In particular any element
$c\in\mathrm{CSpin(V)}$ sends a $\mathbb Z$-graded subalgebra
$\fa=V\oplus S\oplus \fh$ of $\fp$ into an (isomorphic)
$\mathbb Z$-graded subalgebra
$\fa'=c\cdot\fa=V\oplus S\oplus (c\cdot \fh)$ of $\fp$ and, if $\fg$ is
a filtered deformation of $\fa$ associated with $\varphi$ then
$\fg'=c\cdot \fg$ is a filtered deformation of $\fa'$, which is associated
with $\varphi'=c\cdot\varphi$.

We will distinguish $\varphi$ according to whether it is spacelike,
timelike or lightlike and denote by $\Pi \subset V$ the line defined
by the span of $\varphi$. In the first two cases we can decompose
$V = \Pi \oplus \Pi^\perp$ into an orthogonal direct sum and
$\fh_\varphi = \fso(\Pi^\perp) \subset \fso(V)$.
If $\varphi$ is lightlike, we choose an
$\eta$-Witt basis for $V$ such that $V = \RR\left<\be_+,\be_-\right>
\oplus W$ and $\varphi = \be_+ $.  Our plane is $\Pi = \RR\left<\be_+\right>$ and 
$\fh_\varphi = \fso(W) \inplus
  (\be_+ \wedge W) \subset \fso(V)$,
where $\be_+ \wedge W$ is the abelian Lie subalgebra of $\fso(V)$ consisting of null
rotations fixing $\be_+$. In this case we decompose any $v\in V$ into 
\begin{equation*}
  v=v_++v_-+v_\perp\;,
\end{equation*}
where $v_+\in \Pi$, $v_-\in\mathbb R\left<\be_-\right>$ and $v_\perp\in W$.

\begin{lemma}
\label{lem:case2}
Let $[\beta^{\Phi} + \gamma^{\Phi}+\partial\widetilde\psi] \in
  H^{2,2}(\fa_-,\fa)$ be a nontrivial and $\fh$-invariant
  cohomology class with $\varphi\neq 0$ and $a=b=0$. Then $\operatorname{Im}(\gamma^{\Phi})=\fh_{\varphi}$
	and there exists a unique cocycle representative $\beta^{\Phi} + \gamma^{\Phi}+\partial\widetilde\psi$ for which $\gamma^\Phi(s,s)-\widetilde\psi(\kappa(s,s))=0$ for all $s\in S$.
\end{lemma}
\begin{proof}
We already know from (2) of Proposition~\ref{prop:cohgroups} that $\operatorname{Im}(\gamma^{\Phi})\subset\fh_{\varphi}$. 
In addition:
\begin{align}
\eta(w,\gamma^{\Phi}(s,s)v)&=\eta(w,\kappa(s,(v\cdot\varphi+3\varphi\cdot v)\cdot\vol\cdot s))\notag
\\
&=-2\eta(w,\kappa(s,v\cdot\varphi\cdot\vol\cdot s))\notag\\
&=2\left< s, w\cdot v\cdot\imath_{\varphi}\vol\cdot s\right>\notag\\
&=2\left< s, \imath_{w}\imath_{v}(\imath_{\varphi}\vol)\cdot s\right>\label{eq:gammanondeg}\;,
\end{align}
for all $v, w\in V$. Using \eqref{eq:gammanondeg} and
$\odot^2 S=\wedge^1 V\oplus\wedge^2 V$, we first see that
$\gamma^{\Phi}(\wedge^2 V)=0$.  We now break our arguments into two
cases, depending on whether or not the line $\Pi$ corresponding to
$\varphi$ is degenerate.

If $\varphi$ is spacelike or timelike then from \eqref{eq:gammanondeg}
we see that $\gamma^{\Phi}(\Pi)=0$ whereas
\begin{equation*}
  \gamma^{\Phi}|_{\Pi^\perp}:\Pi^\perp\subset\wedge^1 V\longrightarrow\fso(\Pi^\perp)
\end{equation*}
is an $\fso(\Pi^\perp)$-equivariant monomorphism, hence an isomorphism
by dimensional reasons.  If $\varphi$ is lightlike we decompose
\begin{equation}
  \begin{split}
    \eta(w,\gamma^{\Phi}(s,s)v)&=2\left< s, \imath_w\imath_{v}(\imath_{\varphi}\vol)\cdot s\right>\\
    &=2\left< s,
      \imath_{w_\perp}\imath_{v_-}(\imath_{\varphi}\vol)\cdot
      s\right>+2\left< s,
      \imath_{w_-}\imath_{v_\perp}(\imath_{\varphi}\vol)\cdot
      s\right>  + 2\left< s,
      \imath_{w_\perp}\imath_{v_\perp}(\imath_{\varphi}\vol)\cdot s\right>\label{eq:gammadeg}\;,
  \end{split}
\end{equation}
which readily gives $\gamma^{\Phi}(\Pi)=0$, $\gamma^{\Phi}(e_-)$ is a
generator of $\fso(W)$ and, finally, $\gamma^{\Phi}(W)=e_+\wedge W$.
In this case $\gamma^{\Phi}$ is an $\fh_{\varphi}$-equivariant
isomorphism from $\mathbb R\left<\be_-\right>\oplus W$ to
$\fh_{\varphi}$.

To prove the last statement, we recall that
$\gamma^{\Phi}(s,s)-\widetilde\psi(\kappa(s,s))\in\fh$ for all
$s\in S$, by Proposition~\ref{prop:cohgroups}. On the other hand we
just saw that the operator $\gamma^{\Phi}-\widetilde\psi(\kappa(-,-))$
acts trivially on $\wedge^2 V\subset\odot^2 S$ and possibly
non-trivially only on $\wedge^1 V\subset\odot^2 S$. In other words it
is an operator of the form $\psi(\kappa(-,-)):\odot^2 S\to\fh$ for
some $\psi\in C^{2,1}(\fa_-,\fa)=\Hom(V,\fh)$ and such a $\psi$ is
clearly unique, since $\fa$ is fundamental. Subtracting the coboundary
$\partial\psi$ to the cocycle
$\beta^{\Phi} + \gamma^{\Phi}+\partial\widetilde\psi$ gives the last
claim.
\end{proof}

We collect here for later use different equivalent characterizations
of the map $\widetilde\psi:V\to\fso(V)$ associated to the unique
cocycle representative of Lemma~\ref{lem:case2}:
\begin{enumerate}[label=(\roman*)]
\item $\widetilde\psi(\kappa(s,s))=\gamma^{\varphi}(s,s)$ for all $s\in S$;
\item $\widetilde\psi(u)=2\imath_{u}\imath_{\varphi}\vol$ for all $u\in V$;
\item $\widetilde\psi(u)v=2\imath_{v}\imath_{u}\imath_{\varphi}\vol$ for all $u,v\in V$;
\item $\eta(w,\widetilde\psi(u)v)=2\imath_w\imath_{v}\imath_{u}\imath_{\varphi}\vol$ for all $u,v,w\in V$;
\item $\widetilde\psi(u)s=-(\varphi\wedge u)\cdot\vol\cdot s$ for all $u\in V$ and $s\in S$;
\item $(\widetilde\psi(u)v)\cdot s=2(\varphi\wedge u\wedge v)\cdot\vol\cdot s$ for all $u,v\in V$ and $s\in S$.
\end{enumerate}
We also remark that $\widetilde\psi$ is an $\fh_{\varphi}$-equivariant
map with the kernel $\Pi$ and image $\fh_{\varphi}$.

\subsection{Integrability of the infinitesimal deformations}
\label{sec:integr-deform}

In this section we construct a filtered
deformation $\fg$ for any of the nontrivial $\fh$-invariant elements in
$H^{2,2}(\fa_-,\fa)$. Our description of $\fg$ will be very
explicit and rely on a direct check of the Jacobi identities.
To describe the Lie superalgebra structure of $\fg$, it is convenient to introduce a formal parameter $t$ which keeps track of the order of the deformation. In particular, the original graded Lie superalgebra structure on a subalgebra $\fa=V\oplus S\oplus\fh$ of the Poincaré
superalgebra $\fp=V\oplus S\oplus\fso(V)$ has order $t^0$ whereas the infinitesimal deformation has order $t$. 

From Proposition~\ref{thm:filt1}, Proposition~\ref{prop:cohgroups}, and Lemma~\ref{lem:case1}, Lemma~\ref{lem:case2}
we know that there are two different families of non-trivial filtered deformations $\fg$. The first family has $\varphi=0$, $a^2+b^2\neq 0$ and $\fh=\fso(V)$, that is $\fa=\fp$. In this case $\gamma^{\Phi}:\odot^2 S\to\fso(V)$ is surjective and 
by (2) of Proposition~\ref{prop:cohgroups} the filtered Lie superalgebra $\fg$ has the brackets of the form
\begin{equation}
\label{eq:bracket1}
    \begin{aligned}[m]
    [A,v]&=Av+t\lambda(A,v)\\
    [A,s]&=\sigma(A)s\\
    [A,B]&=AB-BA
  \end{aligned}
    \qquad\qquad
  \begin{aligned}[m]
    [v,w] &= t^2\delta(v,w)\\
    [v,s] &= t \beta^{\Phi}(v,s) = tv \cdot (a+b\vol)\cdot s\\
    [s,s] &= \kappa(s,s) + t \gamma^{\Phi}(s,s)~,
    \end{aligned}
\end{equation}
where $A, B\in\fso(V)$, $s\in S$, $v,w\in V$, for some maps
$\lambda:\fso(V)\otimes V\to\fso(V)$ and $\delta:\wedge^2 V\to
\fso(V)$ to be determined. In other words the brackets on $V\otimes S$
and $\odot ^2S$ are respectively given by $\beta^{\Phi}$ and
$\gamma^{\Phi}$ and we can always assume $\alpha=0$ without any loss
of generality.

The second family has $\varphi\neq 0$, $a=b=0$ and $\fh$ is a Lie
subalgebra of the stabiliser
$\fh_{\varphi}=\fso(V)\cap\stab(\varphi)$, see (iii) of
Proposition~\ref{prop:cohgroups}. We recall that
$\varphi\in\wedge^1V$ can be either spacelike, timelike or
lightlike. In this case the bracket on $\odot ^2S$ is simply given by
the Dirac current and the filtered Lie superalgebra $\fg$ has the form
\begin{equation}
\label{eq:bracket2}
  \begin{aligned}[m]
    [A,v]&=Av+t\lambda(A,v)\\
    [s,s] &= \kappa(s,s)\\
    [A,s]&=\sigma(A)s\\
    [A,B]&=AB-BA
  \end{aligned}
  \qquad\qquad
  \begin{aligned}[m]
   [v,w] &=t\alpha(v,w)+t^2\delta(v,w)=t\pi^{\alpha}(\partial\widetilde\psi)(v,w)+t^2\delta(v,w)\\
   &= t\widetilde\psi(v)w-t\widetilde\psi(w)v+t^2\delta(v,w)\\
   [v,s] &=t\beta(v,s)= t \beta^{\Phi}(v,s)+t\pi^{\beta}(\partial\widetilde\psi)(v,s) \\
   &= -\tfrac{1}{2} t(v\cdot\varphi+3\varphi\cdot v)\cdot\vol\cdot s+t\widetilde\psi(v)s~,
  \end{aligned}
\end{equation}
where $A,B\in\fh$, $s\in S$, $v,w\in V$, for some maps $\lambda:\fh\otimes V\to\fh$ and $\delta:\wedge^2 V\to \fh$ to be determined.

To go through all the Jacobi identities
systematically, we use the notation $[ijk]$ for $i,j,k=0,1,2$ to
denote the identity involving $X \in \fa_{-i}$, $Y \in \fa_{-j}$ and
$Z \in \fa_{-k}$. We first consider the second case \eqref{eq:bracket2}, which is slightly more involved, and
claim that the Jacobi identities are satisfied if we set both $\lambda$ and $\delta$ 
to be zero. To show this, it is first convenient to note that $[V,V]\subset V$, $[V,S]\subset S$ and $[S,S]\subset V$
and prove that the putative bracket operations restricted on $V\oplus S$ satisfy  the Jacobi identities. We have:
\begin{itemize}
\item the $[112]$ identity is satisfied by
  virtue of the characterization (i) of $\widetilde\psi$, the $\fh_{\varphi}$-equivariance of the Dirac current
	and the first cocycle condition \eqref{eq:cc1};
	\item the $[111]$ identity is satisfied by
  virtue of the characterization (i) of $\widetilde\psi$
	and the second cocycle condition \eqref{eq:cc2};
\item the $[122]$ identity is satisfied provided
\begin{equation}
\label{eq:Jac122}
[\beta_v,\beta_w]s-\beta_{\alpha(v,w)}s=0\;,
\end{equation}
for all $v, w\in V$ and $s\in S$;
\item the $[222]$ identity is satisfied provided 
\begin{equation}
\label{eq:Jac222}
\mathfrak{S}(\alpha(u,\alpha(v,w)))=0\;,
\end{equation}
where $\mathfrak{S}$ is the cyclic sum on $u,v,w\in V$.
\end{itemize}
Now using characterization (iv) of $\widetilde\psi$ one can check that 
\begin{equation*}
  \eta(x,\widetilde\psi(u)\widetilde\psi(v)w) = 4
  \eta(\imath_{v}\imath_{w}\imath_{\varphi}\vol,\imath_x\imath_u\imath_\varphi\vol)
  = \eta(x,\widetilde\psi(\widetilde\psi(w)v)u)
\end{equation*}
for all $u,v,w,x\in V$, from which
\begin{equation*}
  \begin{split}
    \alpha(u,\alpha(v,w)) &= \widetilde\psi(u) \widetilde\psi(v)w -
    \widetilde\psi(u) \widetilde\psi(w)v     + \widetilde\psi
    (\widetilde\psi(w)v) u - \widetilde\psi (\widetilde\psi(v) w )u\;\\
    &=2\widetilde\psi(u)\widetilde\psi(v)w-2\widetilde\psi(u)\widetilde\psi(w)v\\
    &=4\widetilde\psi(u)\widetilde\psi(v)w
  \end{split}
\end{equation*}
and $\mathfrak{S}(\alpha(u,\alpha(v,w)))=4\mathfrak{S}(\widetilde\psi(u)\widetilde\psi(v)w)=0$ by characterization (iii) of $\widetilde\psi$. This is the [222] Jacobi identity \eqref{eq:Jac222}.
On the other hand, for all $v,w\in V$ and $s\in S$ we have
\begin{equation*}
  \begin{split}
    \beta_v\beta_ws&=-\tfrac{1}{2}\beta_v((w\cdot\varphi+3\varphi\cdot w)\cdot\vol\cdot s)+\beta_v(\widetilde\psi(w)(s))\\
    &=-(\varphi\wedge v-2\eta(\varphi,v))\cdot(\varphi\wedge
    w-2\eta(\varphi,w))\cdot s
    -(\varphi\wedge v-2\eta(\varphi,v))\cdot\vol\cdot\widetilde\psi(w)s\\
    &\;\;\;\;-\vol\cdot\widetilde\psi(v)(\varphi\wedge
    w-2\eta(\varphi,w))\cdot s+\widetilde\psi(v)\widetilde\psi(w)s
  \end{split}
\end{equation*}
and therefore, repeatedly using
equations~\eqref{eq:bivectorproducts} and the fact that
$\widetilde\psi(u)\varphi=0$ for all $u\in V$, also
\begin{equation}
  \label{eq:commbeta}
  \begin{split}
    [\beta_v,\beta_w]s
    &=-[\varphi\wedge v,\widetilde\psi(w)]\vol\cdot s
    +[\varphi\wedge w,\widetilde\psi(v)]\vol\cdot s\\
    &\;\;\;\;-[\varphi\wedge v,\varphi\wedge w]s+[\widetilde\psi(v),\widetilde\psi(w)]s\\
    &=(\varphi\wedge\widetilde\psi(w)v)\cdot\vol\cdot s - (\varphi\wedge\widetilde\psi(v)w)\cdot\vol\cdot s
    \\
    &\;\;\;\;-2\eta(\varphi,\varphi)v\wedge w\cdot s+2\eta(\varphi,v)\varphi\wedge w\cdot s-2\eta(\varphi,w)\varphi\wedge v\cdot s\\
    &\;\;\;\;+[\widetilde\psi(v),\widetilde\psi(w)]s\;.
  \end{split}
\end{equation}
In a similar way we can prove:
\begin{equation}
  \label{eq:beta-alpha}
  \beta_{\alpha(v,w)}s =
  (\varphi\wedge\widetilde\psi(w)v)
  \cdot\vol\cdot s-(\varphi\wedge\widetilde\psi(v)w)
  \cdot\vol\cdot s + \widetilde\psi(\widetilde\psi(v)w)s -
  \widetilde\psi(\widetilde\psi(w)v)s\;.
\end{equation}
In summary we use \eqref{eq:commbeta} and \eqref{eq:beta-alpha}, together with characterizations (v) and (vi) of $\widetilde\psi$, to arrive at:
\begin{equation*}
  \begin{split}
    [\beta_v,\beta_w]s-\beta_{\alpha(v,w)}s&=-2\eta(\varphi,\varphi)v\wedge w\cdot s+2\eta(\varphi,v)\varphi\wedge w\cdot s-2\eta(\varphi,w)\varphi\wedge v\cdot s\\
    &\;\;\;\;+[\widetilde\psi(v),\widetilde\psi(w)]s-\widetilde\psi(\widetilde\psi(v)w)s
    +\widetilde\psi(\widetilde\psi(w)v)s\\
    &=-2\eta(\varphi,\varphi)v\wedge w\cdot s+2\eta(\varphi,v)\varphi\wedge w\cdot s-2\eta(\varphi,w)\varphi\wedge v\cdot s\\
    &\;\;\;\;-[\varphi\wedge v,\varphi\wedge w]s-2\varphi\cdot(\varphi\wedge v\wedge w)s-2(\varphi\wedge v\wedge w)\cdot\varphi\cdot s\\
    &=0\;,
  \end{split}
\end{equation*}
proving the $[122]$ Jacobi identity \eqref{eq:Jac122}.

Let $\fg_-=(V\oplus S,[-,-])$ be the filtered Lie superalgebra
structure on $V\oplus S$ we have just described. Note that the Lie
bracket of $\fg_-$ is defined in terms of $\widetilde\psi$, Clifford
multiplication, Dirac current of spinors and the vector $\varphi$, so
that the stabilizer $\fh_{\varphi}=\fso(V)\cap\stab(\varphi)$ of
$\varphi$ in $\fso(V)$ acts naturally on $\fg_-$ by outer derivations.
It is then clear from \eqref{eq:bracket2} that, for any subalgebra
$\fh$ of $\fh_{\varphi}$, the semidirect sum $\fg=\fh\inplus \fg_-$ is
the required filtered deformation of $\fa=V\oplus S\oplus \fh$.

We now consider the first case \eqref{eq:bracket1} and set $\lambda$ to be zero. We have:
\begin{itemize}
\item the $[000]$ identity is satisfied since $\fso(V)$ is a Lie
  algebra;
\item the $[001]$ and $[002]$ identities are satisfied because
  $S$ and $V$ are $\fso(V)$-modules;
\item the $[011]$ and $[012]$ identities are
  satisfied because the $[SS]$, $[SV]$ Lie brackets are
  $\fso(V)$-equivariant;
\item the $[111]$ identity is satisfied by
  virtue of the second cocycle condition
  \eqref{eq:cc2};
\item the $[022]$ identity requires $\delta:\wedge^2 V\to \fso(V)$ to be $\fso(V)$-equivariant;
\item the $[222]$ identity is satisfied provided
  \begin{equation}
    \label{eq:222coc}
    \mathfrak{S}(\delta(v,w)u)=0\;,
  \end{equation}
  where $\mathfrak{S}$ is the cyclic sum on $v,w,u\in V$;
\item the $[122]$ identity is satisfied provided
  \begin{equation}
    \label{eq:122coc}
    \delta(v,w)s=[\beta^{\Phi}_v,\beta^{\Phi}_w]s\;,
  \end{equation}
  for all $v,w\in v$ and $s\in S$;
\item the $[112]$ identity has a component of order $t$, which is satisfied by virtue 
  of the first cocycle condition \eqref{eq:cc1} and one of order $t^2$, which reads
  \begin{equation}
    \label{eq:112coc}
    \delta(v,\kappa(s,s))=2\gamma^{\Phi}(\beta^{\Phi}_vs,s)\;,
  \end{equation}
  for all $v\in V$ and $s\in S$;
\end{itemize}
Since $\wedge^2 V$ is an irreducible $\fso(V)$-representation of
complex type, we have that the $[022]$ Jacobi identity is satisfied if
and only if there exist $r, r'\in\mathbb R$ such that
\begin{equation*}
  \delta(v,w)u=r(\eta(v,u)w-\eta(w,u)v)+r'\star(v\wedge w\wedge u)\;,  
\end{equation*}
for all $v,w,u\in V$. However it is easy to see that \eqref{eq:222coc}
implies $r'=0$.

We will now show that \eqref{eq:122coc} and \eqref{eq:112coc} hold
true for $r=4(a^2+b^2)$. Indeed:
\begin{equation*}
  \begin{split}
    \delta(v,w)s&=\tfrac{r}{4}(v\cdot w\cdot s-w\cdot v\cdot s)\;,\\
    [\beta^{\Phi}_v,\beta^{\Phi}_w] s&=v\cdot (a+b\vol)\cdot w\cdot
    (a+b\vol)\cdot s - w\cdot (a+b\vol)\cdot v\cdot (a+b\vol)\cdot s\\
    &= (a^2+b^2)(v\cdot w\cdot s-w\cdot v\cdot s)\;,
\end{split}
\end{equation*}
for all $v,w\in V$, $s\in S$, whereas 
\begin{equation*}
  \begin{split}
    \eta(\delta(v,\kappa(s,s))u,w)&=r(\eta(v,u)\eta(\kappa(s,s),w)-\eta(\kappa(s,s),u)\eta(v,w))\\
    &=r(\eta(v,u)\left< s,w\cdot s\right>-\eta(v,w)\left< s,u\cdot s\right>)\;,\\
    2\eta(\gamma^{\Phi}(\beta_v^{\Phi}s,s)u,w)&=-2\eta(\kappa(\beta_v^{\Phi}s,\beta^{\Phi}_u s),w)-2\eta(\kappa(s,\beta^{\Phi}_u\beta_v^{\Phi}s),w)\\
    &=2(a^2+b^2)(\left< s,v\cdot w\cdot u\cdot s\right>-\left< s,w\cdot u\cdot v\cdot s\right>)\\
    &=4(a^2+b^2)(\eta(v,u)\left< s,w\cdot s\right>-\eta(v,w)\left< s,u\cdot s\right>)\;,
  \end{split}
\end{equation*}
for all $v,w,u\in V$, $s\in S$.

\subsection{Summary}
\label{sec:summary}

We summarise the results of Sections~\ref{sec:preliminariesA},
\ref{sec:cohomologyA} and \ref{sec:integr-deform} in the following

\begin{theorem}
  \label{thm:final}
  There are exactly two families of nontrivial filtered deformations
  $\fg=\fg_{\bar 0}\oplus \fg_{\bar 1}$ of $\mathbb Z$-graded
  subalgebras $\fa=V\oplus S\oplus \fh$ of the Poincaré superalgebra
  $\fp=V\oplus S\oplus \fso(V)$, which we now detail:
  \begin{enumerate}
  \item In this case $\fh=\fso(V)$, there exist $a,b\in\mathbb R$ such
    that $a^2+b^2\neq 0$ and the Lie brackets of $\fg$ are given by
    \begin{equation}
      \label{eq:thm-filtdef1}
      \begin{aligned}[m]
        [A,v]&=Av\\
        [A,s]&=\sigma(A)s\\
        [A,B]&=AB-BA
      \end{aligned}
      \qquad\qquad
      \begin{aligned}[m]
        [v,w] &= 4(a^2+b^2)v\wedge w\\
        [v,s] &=v \cdot (a+b\vol)\cdot s\\
        [s,s] &= \kappa(s,s) + \gamma^{(a,b)}(s,s)~,
      \end{aligned}
    \end{equation}
    where $v,w\in V$, $s\in S$, $A,B\in\fso(V)$ and
    $\gamma^{(a,b)}(s,s)\in\fso(V)$ is defined by
    \begin{equation*}
      \gamma^{(a,b)}(s,s)v=-2\kappa(s, v\cdot(a+b\vol)\cdot s)\;;
    \end{equation*}
  \item In this case there exists a nonzero $\varphi\in V$, $\fh$ is
    any Lie subalgebra of the stabiliser
    $\fh_{\varphi}=\fso(V)\cap\stab(\varphi)$ of $\varphi$ in
    $\fso(V)$ and the Lie brackets of $\fg$ are given by
    \begin{equation}
      \label{eq:thm-filtdef2}
      \begin{aligned}[m]
        [A,v]&=Av\\
        [A,s]&=\sigma(A)s\\
        [A,B]&=AB-BA
      \end{aligned}
      \qquad\qquad
      \begin{aligned}[m]
        [v,w]&= \widetilde\psi(v)w-\widetilde\psi(w)v\\
        [v,s] &=-\frac{1}{2}(v\cdot\varphi+3\varphi\cdot v)\cdot\vol\cdot s+\widetilde\psi(v)s\\
        [s,s] &= \kappa(s,s)~,
      \end{aligned}
    \end{equation}
    where $v,w\in V$, $s\in S$, $A,B\in\fh$ and
    $\widetilde\psi(v)\in\fh_{\varphi}$ is defined by
    $\widetilde\psi(v)=2\imath_v\imath_\varphi\vol$. In particular
    $\fg_-=V\oplus S$ is an ideal of $\fg$ and $\fg=\fh\inplus \fg_-$ is
    the semidirect sum of $\fh$ and $\fg_-$ ($\fh$ acts on $\fg_-$ by
    restricting the vector and spinor representations of $\fso(V)$).
  \end{enumerate}
  Note that the associated homogeneous Lorentzian manifolds $(M=G/H,g)$,
  $Lie(G)=\fg_{\bar 0}$, $Lie(H)=\fh$  always admit a reductive
  decomposition $\fg_{\bar 0}=\fh\oplus V$. In the first family $(M,g)$
  is locally isometric to $\mathrm{AdS}_4$ whereas in the second family
  the geometry is that of a Lie group with a bi-invariant metric, more
  precisely:
  \begin{enumerate}[label=(\roman*)]
  \item If $\varphi$ is spacelike then $\fh_{\varphi}\simeq \fso(1,2)$
    and $(M,g)$ is locally isometric to $\mathrm{AdS}_3\times\mathbb R$;
  \item If $\varphi$ is timelike then $\fh_{\varphi}\simeq \fso(3)$ and
    $(M,g)$ is locally isometric to $\mathbb R\times \Sph^3$;
  \item If $\varphi$ is lightlike then  $\fh_{\varphi}\simeq
    \fso(2)\inplus \mathbb R^2$ and we have the so-called Nappi-Witten
    group \cite{NW}, a central extension of the Lie group of Euclidean
    motions of the plane. Explicitly, if we choose an $\eta$-Witt basis
    for $V$ with $\varphi=e_{+}$, then the only nonzero Lie brackets of
    the Lie algebra of the Nappi-Witten group are:
    \begin{equation*}
      [e_-,e_1]=4e_2\;,\qquad[e_-,e_2]=-4e_1\;,\qquad[e_1,e_2]=-4e_+\;.
    \end{equation*}
  \end{enumerate}
\end{theorem}

\section{Conclusions}
\label{sec:conclusions}

In this paper, we have considered the supersymmetries of rigid supersymmetric
field theories on Lorentzian four-manifolds from the viewpoint of their Killing superalgebras.

We showed that the relevant Killing spinor equations, which we
identify with the defining condition for bosonic supersymmetric
backgrounds of minimal off-shell supergravity in four dimensions,
admit a cohomological interpretation in terms of the Spencer group
$H^{2,2}(\fp_-,\fp)$ of the $N{=}1$ Poincaré superalgebra
$\fp$ in four dimensions. This result is in analogy with a similar
result in eleven dimensions
\cite{Figueroa-O'Farrill:2015efc,Figueroa-O'Farrill:2015utu}.

We then gave a self-contained proof of the fact that supergravity
Killing spinors generate a Lie superalgebra, and that this Lie
superalgebra is a filtered subdeformation of $\fp$. Finally we
classified, up to local isometry, the geometries admitting the maximum
number of Killing spinors: Minkowski space, $\AdS_4$ and the
nonabelian Lie groups with a Lorentzian bi-invariant metric, namely
$\AdS_3\times\RR$, $\RR \times \Sph^3$ and the Nappi--Witten group
$\NW_4$. Our approach here is based on two independent arguments. In
Section~\ref{sec:zero-curvature-eqns} we solved the flatness equations
for the connection defining the Killing spinor equations and described
the corresponding Lorentzian geometries. In
Section~\ref{sec:maxim-fil-def} we used again Spencer cohomology
techniques to describe the filtered subdeformations of $\fp$ with
maximum odd dimension and recovered in this way the Killing
superalgebras of the maximally supersymmetric backgrounds.

None of the geometries in our classification are new.  The novelty in
this paper lies in our approach, which \emph{systematises} the
search for backgrounds on which one can define rigid supersymmetric
field theories by mapping it to an algebraic problem on which we can
bring to bear representation-theoretic techniques.  In forthcoming
work, we shall apply these techniques to a broader class of field theories with rigid supersymmetry in higher dimensions.

\section*{Acknowledgments}

The research of JMF is supported in part by the grant ST/L000458/1
``Particle Theory at the Tait Institute'' from the UK Science and
Technology Facilities Council.  That grant also funded a visit of PdM
to Edinburgh to start this collaboration.  The research of AS is fully
supported by a Marie-Curie research fellowship of the ''Istituto
Nazionale di Alta Matematica'' (Italy).  We are grateful to these
funding agencies for their support.

\appendix

\section{Conventions and spinorial algebraic identities}
\label{sec:conventions}

In this appendix we define our conventions for Clifford algebras,
spinors and derive a number of useful algebraic identities we will
have ample opportunity to apply in the bulk of the paper.

\subsection{Clifford algebra conventions}
\label{sec:cliff-algebra-conv}

Let $(V,\eta)$ be a four-dimensional Lorentzian vector space, by which
we mean that $\eta$ has signature $-2$ (``mostly minus'').  We may
choose an $\eta$-orthonormal basis $\be_\mu =
(\be_0,\be_1,\be_2,\be_3)$ with $\eta_{\mu\nu} = \eta(\be_\mu,\be_\nu) =
\diag(+1,-1,-1,-1)$.  Such a basis defines an isomorphism $(V,\eta)
\cong \RR^{1,3}$.

The Clifford algebra $\Cl(V)$ associated to $(V,\eta)$ is the real,
associative, unital algebra generated by $V$ (and the identity $\1$)
subject to the Clifford relation (please notice the sign!)
\begin{equation}
  \label{eq:clifford}
  v^2 = - \eta(v,v) \1 \qquad\forall~v\in V~.
\end{equation}

As a vector space, $\Cl(V) \cong \Lambda V = \bigoplus_{p=0}^4
\wedge^p V$.  If $v \in V$ and $\phi \in \wedge^p V$, their Clifford
product, denoted by $\cdot$, is given by
\begin{equation}
  \label{eq:cliffordproduct}
  v \cdot \phi = v \wedge \phi - \iota_{v^\flat} \phi~,
\end{equation}
where $v^\flat \in V^*$ is the dual covector defined by the inner
product: $v^\flat(w) = \eta(v,w)$, for all $w \in V$.  We will often
drop the superscript $\flat$ if it is unambiguous to do so.  The Clifford
algebra is not commutative:
\begin{equation}
  \label{eq:productclifford}
  \phi \cdot v = (-1)^p \left( v \wedge \phi + \iota_v \phi\right)~.
\end{equation}
Continuing in this way we may derive the Clifford product of $\phi \in
\wedge^p V$ with bivectors:
\begin{equation}
  \label{eq:bivectorproducts}
  \begin{split}
    (v \wedge w) \cdot \phi &= v \wedge w \wedge \phi + \iota_v \iota_w 
    \phi - v \wedge \iota_w \phi + w \wedge \iota_v \phi\\
    (v \wedge w) \cdot \phi &= v \wedge w \wedge \phi + \iota_v \iota_w 
    \phi + v \wedge \iota_w \phi - w \wedge \iota_v \phi~.
  \end{split}
\end{equation}

Let us introduce the volume element $\vol = \be_0 \wedge \be_1 \wedge \be_2
\wedge \be_3 \in \wedge^4 V$.  It obeys
\begin{equation*}
  \vol^2 = -\1 \qquad\text{and}\qquad \vol \cdot \phi = (-1)^p \phi
  \cdot \vol~,
\end{equation*}
for $\phi \in \wedge^p V$.  In particular, it is not central.
Clifford multiplication by the volume element agrees (up to a sign)
with Hodge duality:
\begin{equation}
  \label{eq:Hodge}
  \star \phi = (-1)^{p(p+1)/2} \phi \cdot \vol~,
\end{equation}
for $\phi \in \wedge^pV$.  It follows that $\star^2 = (-1)^{p+1}$ on
$\wedge^p V$.  In particular, it is a complex structure on bivectors,
as expected.

The Lie algebra $\fso(V)$ of $\eta$-skewsymmetric endomorphisms of $V$
is isomorphic, as a vector space, to $\wedge^2 V$.  If $v \wedge w \in
\wedge^2 V$, then the corresponding endomorphism is defined by
\begin{equation}
  \label{eq:soVasL2}
  (v \wedge w) (u) = \iota_{u^\flat} (v \wedge w) = \eta(u,v) w -
  \eta(u,w) v ~.
\end{equation}
We embed $\fso(V)$ in $\Cl(V)$ by sending
\begin{equation}
  \label{eq:soVinClV}
  v \wedge w \mapsto \tfrac14 [v,w] = \tfrac14 (v \cdot w - w \cdot v)~.
\end{equation}
Indeed, one checks that the Clifford commutator
\begin{equation*}
   \left[\tfrac14 [v,w],u\right] = \eta(u,v) w - \eta(u,w) v~,
\end{equation*}
agrees with equation~\eqref{eq:soVasL2}.

\subsection{Clifford module conventions}
\label{sec:cliff-module-conv}

The Clifford algebra $Cl(V)$ is isomorphic, as a real associative
algebra, to the algebra $\Mat_4(\RR)$ of $4\times 4$ real matrices.
Being simple, this algebra has a unique (up to isomorphism) nontrivial
irreducible module, which is real and four-dimensional.  Let $S$
denote the unique (up to isomorphism) irreducible $\Cl(V)$-module, so
that $\Cl(V) \cong \End S$.  Restricting to $\fso(V) \subset \Cl(V)$,
we obtain a representation $\sigma$ of $\fso(V)$ on $S$:
\begin{equation}
  \label{eq:spinrep}
  \sigma(v \wedge w) s = \tfrac14 [v,w] \cdot s~.
\end{equation}

On $S$ we have a symplectic structure $\left<-,-\right>$ realising one
of the canonical anti-involutions of the  Clifford algebra:
\begin{equation}
  \label{eq:symplectic}
  \left<v \cdot s_1, s_2 \right> = -  \left<s_1, v\cdot s_2 \right>~,
\end{equation}
for all $v \in V$ and $s_1,s_2 \in S$.  It follows that it is also
$\fso(V)$-invariant:
\begin{equation}
  \left<\sigma(A) s_1, s_2\right> = - \left<s_1, \sigma(A) s_2\right>~,
\end{equation}
for all $A \in \fso(V)$.  More generally, it follows from repeated
application of equation~\eqref{eq:symplectic}, that if $\phi \in
\wedge^p V$, then
\begin{equation}
  \label{eq:symmetry}
  \left<\phi \cdot s_1, s_2 \right> = (-1)^{p(p+1)/2}  \left<s_1,
    \phi\cdot s_2 \right>~.
\end{equation}
We can therefore decompose $\End S \cong \odot^2 S \oplus \wedge^2 S$
into representations of $\fso(V)$ as
\begin{equation}
  \label{eq:EndSasLambda}
  \odot^2 S \cong  \wedge^1 V \oplus \wedge^2 V  \cong V \oplus
  \fso(V) \qquad\text{and}\qquad
  \wedge^2 S \cong \wedge^0V \oplus \wedge^3 V \oplus \wedge^4 V
  \cong 2 \RR \oplus V~.
\end{equation}

Associated with $s \in S$ there is a vector $\kappa$, called the
\emph{Dirac current} of $s$, that is defined by
\begin{equation}
  \label{eq:DiracCurrent}
  \eta(\kappa,v) = \left<s, v\cdot s\right>~,
\end{equation}
for all $v \in V$.  There is also a \emph{Dirac $2$-form}
$\omega^{(2)}$ defined by
\begin{equation}
  \label{eq:Dirac2form}
  \omega^{(2)}(v,w) = \left<s, v \cdot w \cdot s\right>~.
\end{equation}
(One checks that indeed $\omega^{(2)}(v,w) = - \omega^{(2)}(w,v)$.)
In addition we have a second $2$-form $\widetilde\omega^{(2)}$ and a
$3$-form $\omega^{(3)}$ defined by
\begin{equation}
  \label{eq:otherforms}
  \widetilde\omega^{(2)}  (v,w) = \left<s, v \cdot w \cdot \vol \cdot
    s\right> \qquad\text{and}\qquad \omega^{(3)}(u,v,w) = \left<s, u
    \cdot v \cdot w \cdot \vol \cdot s\right>~.
\end{equation}
It follows that
\begin{equation}
  \widetilde\omega^{(2)} = - \star \omega^{(2)} \qquad\text{and}\qquad
  \omega^{(3)} = - \star \omega^{(1)}~,
\end{equation}
where $\omega^{(1)} = \kappa^\flat$ is the one-form dual to the Dirac current.

\subsubsection{Gamma matrices}
\label{sec:gamma-matrices}

We denote the endomorphism of $S$ corresponding to $\be_\mu \in V$ by
$\Gamma_\mu$ and note that the Clifford relation \eqref{eq:clifford} turns into
the well-known
\begin{equation}
  \label{eq:gammamatrices}
  \Gamma_\mu \Gamma_\nu +   \Gamma_\nu \Gamma_\mu = - 2 \eta_{\mu\nu} \1~,
\end{equation}
where we let $\1$ denote also the identity endomorphism of $S$.  The
vector space isomorphism  $\Cl(V) \cong \Lambda V$ defines a vector
space isomorphism $\End S \cong \Lambda V$ and this in turns defines
the standard $\RR$-basis of $\End S$:
\begin{equation}
  \1 \qquad \Gamma_\mu \qquad \Gamma_{\mu\nu} \qquad \Gamma_\mu
  \Gamma_5 \qquad \Gamma_5~,
\end{equation}
where we have introduced
$\Gamma_5=\Gamma_0 \Gamma_1 \Gamma_2 \Gamma_3$ as the endomorphism
corresponding to the volume element and
$\Gamma_{\mu\nu} = \tfrac12 [\Gamma_\mu,\Gamma_\nu]$.  In the same way
we define the totally skewsymmetric products $\Gamma_{\mu\nu\rho}$ and
$\Gamma_{\mu\nu\rho\sigma}$, which obey
\begin{equation}
  \label{eq:duality}  
    \Gamma_{\mu\nu\rho} = \epsilon_{\mu\nu\rho\sigma}
    \Gamma^\sigma\Gamma_5 \qquad \qquad
    \Gamma_{\mu\nu} \Gamma_5 = \tfrac12
    \epsilon_{\mu\nu\rho\sigma} \Gamma^{\rho\sigma} \qquad\qquad
    \Gamma_5 = - \tfrac1{4!} \epsilon_{\mu\nu\rho\sigma} \Gamma^{\mu\nu\rho\sigma}~,
\end{equation}
where $\epsilon_{0123} = +1$, we raise and lower indices with $\eta$
and where the Einstein summation convention is in force.  Some useful
identities involving $\epsilon_{\mu\nu\rho\sigma}$ are
\begin{equation}
  \label{eq:epseps}
  \tfrac16 \epsilon_{\mu\nu\rho\sigma} \epsilon^{\alpha\nu\rho\sigma} = - \delta_\mu^\alpha \qquad \tfrac12 \epsilon_{\mu\nu\rho\sigma} \epsilon^{\alpha\beta\rho\sigma} = - \left(\delta_\mu^\alpha
    \delta_\nu^\beta - \delta_\mu^\beta \delta_\nu^\alpha \right)~,
\end{equation}
and
\begin{equation}
  \label{eq:epseps2}
  \epsilon_{\mu\nu\rho\sigma} \epsilon^{\alpha\beta\gamma\sigma} = -
  \left(\delta_\mu^\alpha \delta_\nu^\beta \delta_\rho^\gamma - \delta_\mu^\alpha \delta_\nu^\gamma \delta_\rho^\beta + \delta_\mu^\beta \delta_\nu^\gamma \delta_\rho^\alpha - \delta_\mu^\beta \delta_\nu^\alpha \delta_\rho^\gamma +\delta_\mu^\gamma \delta_\nu^\alpha \delta_\rho^\beta - \delta_\mu^\gamma \delta_\nu^\beta \delta_\rho^\alpha \right)~,
\end{equation}
whereas some useful trace-like identities involving the $\Gamma_\mu$
are
\begin{equation}
  \label{eq:traces}
  \Gamma^\nu \Gamma_\mu \Gamma_\nu = 2 \Gamma_\mu
  \qquad\text{and}\qquad
  \Gamma^\rho \Gamma_{\mu\nu} \Gamma_\rho = 0~.
\end{equation}

Let $A \in \fso(V)$ be an $\eta$-skewsymmetric endomorphism of $V$.
Its matrix relative to an $\eta$-orthonormal basis $\be_\mu$ has
entries $A^\nu{}_\mu$ defined by
\begin{equation}
  A \be_\mu = \be_\nu A^\nu{}_\mu~,
\end{equation}
whose corresponding skew-symmetric bilinear form has entries
\begin{equation}
  \eta(\be_\nu, A \be_\mu) = A_{\nu\mu}~.
\end{equation}
This in turn gives rise to a bivector
$\tfrac12 A^{\nu\mu} \be_\mu \wedge \be_\nu$ and the map
$\fso(V) \to \wedge^2V$ thus defined is the inverse to the one in
equation~\eqref{eq:soVasL2}.  From equation~\eqref{eq:spinrep}, we see
that the spin representation $\sigma : \fso(V) \to \End S$ sends $A$
to
\begin{equation}
  \sigma(A) = \tfrac14 A^{\nu\mu} \Gamma_{\mu\nu} = - \tfrac14
  A^{\mu\nu} \Gamma_{\mu\nu}~.
\end{equation}

It is often convenient to introduce the notation $\sbar_1 s_2 =
\left<s_1, s_2\right>$ and hence to write the components of the Dirac
current and the Dirac $2$-form as
\begin{equation}
  \kappa^\mu = \sbar \Gamma^\mu s \qquad\text{and}\qquad
  \omega^{(2)}_{\mu\nu} = \sbar \Gamma_{\mu\nu} s~,
\end{equation}
and similarly for their (negative) duals
\begin{equation}
  \widetilde\omega^{(2)}_{\mu\nu} = \sbar \Gamma_{\mu\nu} \Gamma_5 s
  \qquad\text{and}\qquad
  \omega^{(3)}_{\mu\nu\rho} = \sbar \Gamma_{\mu\nu\rho}\Gamma_5 s~,
\end{equation}
which, using the relations \eqref{eq:duality}, can be expressed as
\begin{equation}
  \widetilde\omega^{(2)}_{\mu\nu} = \tfrac12
  \epsilon_{\mu\nu\rho\sigma} \omega^{(2)}{}^{\rho\sigma} \qquad\text{and}\qquad
  \omega^{(3)}_{\mu\nu\rho} = - \epsilon_{\mu\nu\rho\sigma} \kappa^\sigma~.
\end{equation}

\subsection{Spinorial identities}
\label{sec:spinorial-identities}

Let $s_1,s_2 \in S$.  The rank-one endomorphism $s_2\sbar_1$ defined
by $(s_2\sbar_1)(s) = ( \sbar_1 s ) s_2$ can be expressed in terms of
the standard basis for $\End S$ via the \emph{Fierz identity}
\begin{equation}
  \label{eq:Fierz12}
  s_2 \sbar_1 = \tfrac14 \left( ( \sbar_1 s_2 ) \1 -
    ( \sbar_1\Gamma^\mu s_2 ) \Gamma_\mu - \tfrac12
    ( \sbar_1\Gamma^{\mu\nu} s_2 ) \Gamma_{\mu\nu} -
    ( \sbar_1\Gamma^\mu\Gamma_5 s_2 ) \Gamma_\mu \Gamma_5 -
    ( \sbar_1 \Gamma_5 s_2 ) \Gamma_5 \right)~,
\end{equation}
which specialises when $s_1 = s_2 = s$ to
\begin{equation}
  \label{eq:Fierz}
  s\sbar = -\tfrac14 \kappa - \tfrac14 \omega^{(2)} = -\tfrac14 \left(\kappa^\mu \Gamma_\mu + \tfrac12
    \omega^{(2)}_{\mu\nu} \Gamma^{\mu\nu} \right)~.
\end{equation}

There are a number of algebraic identities relating a spinor $s$, its
Dirac current and Dirac $2$-form and their duals, which are collected in
the following

\begin{proposition}
  \label{prop:clifford}
  Let $s \in S$ and $\kappa$ be its Dirac current, $\omega^{(1)} =
  \kappa^\flat$, $\omega^{(2)}$ its Dirac $2$-form, 
  $\widetilde\omega^{(2)}= - \star \omega^{(2)}$ and
  $\omega^{(3)} = - \star (\kappa^\flat)$.  Then the following
  identities hold:
  \begin{multicols}{2}
    \begin{enumerate}[label=(\alph*)]
    \item $\kappa \cdot s = 0$
    \item $\omega^{(2)} \cdot s = 0$
    \item $\widetilde \omega^{(2)} \cdot s = 0$
    \item $\omega^{(3)} \cdot s = 0$
    \item $\eta(\kappa,\kappa) = 0$
    \item $\left(\omega^{(2)},\omega^{(2)}\right)_\eta = 0$
    \item $\left(\widetilde\omega^{(2)},\widetilde\omega^{(2)}\right)_\eta = 0$
    \item $\left(\omega^{(3)},\omega^{(3)}\right)_\eta = 0$
    \item $\iota_\kappa \omega^{(2)} = 0$
    \item $\iota_\kappa \widetilde\omega^{(2)} = 0$
    \item $\iota_\kappa \omega^{(3)} = 0$
    \item $\omega^{(1)} \wedge \omega^{(2)} = 0$
    \item $\omega^{(1)} \wedge \widetilde\omega^{(2)} = 0$
    \item $\omega^{(1)} \wedge \omega^{(3)} = 0$
    \end{enumerate}
  \end{multicols}
\end{proposition}

\begin{proof}
  \begin{enumerate}[label=(\alph*)]
  \item This is equivalent to $\kappa^\rho \Gamma_\rho s = 0$.  Using
    the Fierz identity~\eqref{eq:Fierz},
    \begin{equation*}
      \begin{split}
        \kappa^\rho \Gamma_\rho s  &= \Gamma_\rho s \sbar \Gamma^\rho s \\
        &= \Gamma_\rho \left( -\tfrac14 \left(\sbar \Gamma^\mu s
            \Gamma_\mu + \tfrac12 \sbar \Gamma^{\mu\nu} s
            \Gamma_{\mu\nu} \right)\right) \Gamma^\rho s\\
        &= -\tfrac14 ( \sbar \Gamma^\mu s ) \Gamma_\rho \Gamma_\mu
        \Gamma^\rho s -\tfrac18 ( \sbar \Gamma^{\mu\nu} s ) \Gamma_\rho
        \Gamma_{\mu\nu} \Gamma^\rho s\\
        &= -\tfrac12 ( \sbar \Gamma^\mu s ) \Gamma_\mu s\\
        &= -\tfrac12 \kappa^\mu \Gamma_\mu s~,
      \end{split}
    \end{equation*}
    where we have used the trace identities~\eqref{eq:traces}.

  \item Using that $\sbar s = 0$, we see from the Fierz
    identity~\eqref{eq:Fierz} and part (a) that
    \begin{equation*}
      \omega^{(2)}_{\mu\nu} \Gamma^{\mu\nu} s = 0~.
    \end{equation*}

  \item This follows from $\widetilde\omega^{(2)} = \vol \cdot \omega^{(2)}$ and
    part (b).

  \item This follows from $\omega^{(3)} = - \vol \cdot \kappa$ and part (a).

  \item From (a) it follows that $\kappa$ is null:
    \begin{equation*}
      \eta(\kappa,\kappa) = \left<s, \kappa \cdot s \right>= 0~.
    \end{equation*}

  \item Similarly, from (b) it follows that $\omega^{(2)}$ is null:
    \begin{equation*}
      \left(\omega^{(2)},\omega^{(2)}\right)_\eta = \left<s, \omega^{(2)} \cdot s
      \right> = 0~.
    \end{equation*}

  \item This follows from the fact that $\omega^{(2)}$ is null and that
    Hodge duality is an isometry (up to sign).

  \item This follows from the fact that $\kappa$ is null and that
    Hodge duality is an isometry (up to sign).

  \item This is equivalent to $\omega^{(2)}(\kappa, v) = 0$ for all $v$, but
    \begin{equation*}
      \omega^{(2)}(v,\kappa) = \left<s, v \cdot \kappa \cdot s\right> = 0~,
    \end{equation*}
    where we have used (a) above.

  \item This follows from (a) and
    \begin{equation*}
      \widetilde\omega^{(2)}(v,\kappa) = \left<s, v \cdot \kappa \cdot \vol
        \cdot s \right>= - \left<s, v \cdot \vol \cdot \kappa \cdot  s
      \right> = 0~.
    \end{equation*}

  \item Again this follows from (a) and
    \begin{equation*}
      \widetilde\omega^{(2)}(u,v,\kappa) = \left<s, u\cdot v \cdot \kappa \cdot \vol
        \cdot s \right>= - \left<s, u \cdot v \cdot \vol \cdot \kappa \cdot  s
      \right> = 0~.
    \end{equation*}

  \item We prove the equivalent statement $\star (\omega^{(1)} \wedge \omega^{(2)}) = 0$:
    \begin{equation*}
      \begin{split}
        \epsilon_{\mu\nu\rho\sigma} \kappa^\nu \omega^{(2)}{}^{\rho\sigma} &=
        \epsilon_{\mu\nu\rho\sigma} \kappa^\nu \sbar
        \Gamma^{\rho\sigma}s\\
        &= 2 \kappa^\nu \sbar \Gamma_{\mu\nu} \Gamma_5 s\\
        &=2 \kappa^\nu \sbar \Gamma_\mu \Gamma_\nu \Gamma_5 s\\
        &=- 2 \sbar \Gamma_\mu \Gamma_5 \kappa^\nu \Gamma_\nu s = 0~,
      \end{split}
    \end{equation*}
    again using (a) above.

  \item Similar to the previous part, we prove that $\star (\omega^{(1)}
    \wedge \omega^{(2)}) = 0$:
    \begin{equation*}
      \begin{split}
        \epsilon_{\mu\nu\rho\sigma} \kappa^\nu \widetilde\omega^{(2)}{}^{\rho\sigma} &=
        \epsilon_{\mu\nu\rho\sigma} \kappa^\nu \sbar
        \Gamma^{\rho\sigma}\Gamma_5 s\\
        &= -2 \kappa^\nu \sbar \Gamma_{\mu\nu} s\\
        &=- 2 \kappa^\nu \sbar \Gamma_\mu \Gamma_\nu s = 0~,
      \end{split}
    \end{equation*}
    again using (a).

  \item By definition of Hodge star and (e) above,
    \begin{equation*}
      \omega^{(1)}\wedge \omega^{(3)} = - \omega^{(1)} \wedge \star
      \omega^{(1)} = - \eta(\omega^{(1)},\omega^{(1)}) \vol = 0~.
    \end{equation*}
  \end{enumerate}

\end{proof}

Two remarks are worth mentioning.  The first is that from parts (l),
(m) and (n) in the above proposition, it follows that $\omega^{(2)} = \omega^{(1)}
\wedge \theta$, $\widetilde\omega^{(2)} = \omega^{(1)} \wedge \widetilde\theta$
and $\omega^{(3)} = \omega^{(1)} \wedge \theta^{(2)}$ for some covectors
$\theta,\widetilde\theta$ and $2$-form $\theta^{(2)}$ which are
defined only modulo the ideal generated by $\omega^{(1)}$.

A second remark is that it is possible to prove the above proposition
without resorting to the Fierz identity, by exploiting the
representation theory of the spin group. The group $\Spin(V)$ sits
inside the Clifford algebra $\Cl(V)$ and hence $S$ becomes a
$\Spin(V)$-module by restriction. The volume element defines a complex
structure on $S$ which is invariant under the spin group. The identity
component of the spin group is isomorphic to $\SL(2,\CC)$ under which
$S$ is the fundamental $2$-dimensional complex representation. The
orbit structure of $S$ under $\Spin(V)$ is therefore very simple;
namely, there are two orbits: a degenerate orbit consisting of the
zero spinor and an open orbit consisting of all the nonzero spinors.
The stabiliser of a nonzero spinor $s$ is the abelian subgroup $H_s$
consisting of the null rotations in the direction of its Dirac current
$\kappa$, and it is a subgroup of the stabiliser of any object we can
construct from $s$ in a $\Spin(V)$-equivariant fashion: e.g., the
Dirac current and the Dirac $2$-form. Now the $H_s$-invariant
$2$-forms can be seen to be of the form $\kappa \wedge \theta$, for
some ``transverse'' $1$-form $\theta$, and hence the Dirac $2$-form
$\omega^{(2)}$ has this form. By equivariance under $H_s < \Spin(V)$,
the Clifford product of $\omega^{(2)}$ on $s$ must be again
proportional to $s$, but by squaring we see that the constant of
proportionality must be zero. Finally, Clifford multiplication by the
spacelike $\theta$ is invertible, so it must be that $\kappa$
Clifford-annihilates $s$.

\subsection{A further property of the Dirac current}
\label{sec:prop-dirac-curr}

For completeness we discuss a further algebraic properties of the
Dirac current. Recall that if $s \in S$, its Dirac current $\kappa$ is
defined by equation~\eqref{eq:DiracCurrent}. Let us define a symmetric
bilinear map $\kappa: S \otimes S \to V$ by
\begin{equation}
  \kappa(s_1,s_2) = \tfrac12 \left( \kappa_{s_1+s_2} - \kappa_{s_1} -
    \kappa_{s_2}\right),
\end{equation}
where $\kappa_s$ denotes the Dirac current of $s$.  It follows from
the representation theory of $\fso(V)$ that the map $\kappa$ is
surjective onto $V$.  Now consider a linear subspace $S'\subset S$ and
let $V'\subset V$ denote the image of the map $\kappa$ restricted to
$S' \otimes S'$.  For which $S'$ do we still have that $V'=V$?  The
following lemma, which is a modification of the similar result in
\cite{FigueroaO'Farrill:2012fp} for eleven dimensions, shows that this
holds provided $\dim S'> 2$.

\begin{lemma}
  \label{lem:homogeneity}
  Let $S'\subset S$ be a linear subspace with $\dim S' > \tfrac12 \dim
  S$.  Then the restriction of $\kappa$ to $S'\otimes S'$ is surjective
  onto $V$.
\end{lemma}

\begin{proof}
  Let $S' \subset S$ have $\dim S' > \tfrac12 \dim S$.  Let $V' = \im
  \left.\kappa\right|_{S'\otimes S'}$ and let $v \in (V')^\perp$.  We
  want to show that $v=0$ so that $(V')^\perp = 0$ and hence $V'= V$.
  By definition, $v$ is perpendicular to $\kappa(s_1,s_2)$ for all
  $s_1,s_2 \in S'$; equivalently, $\left<s_1, v \cdot s_2\right> =
  0$.  This means that Clifford multiplication by $v$ maps $S'\to
  (S')^\perp$, where $\perp$ here means the symplectic perpendicular.
  Because of the hypothesis on the dimension of $S'$, $\dim (S')^\perp
  < \dim S'$, so that Clifford multiplication by $v$ has nontrivial
  kernel.  By the Clifford relation \eqref{eq:clifford}, it follows
  that $v$ is null.  In other words, every vector in $(V')^\perp$ is
  null, and this means that $\dim (V')^\perp \leq 1$.  Now for every
  $s \in S'$, $\kappa(s,s)$ is null and perpendicular to the null
  vector $v$, so that one of two situations must occur: either $v=0$
  or else $\kappa(s,s)$ is collinear with $v$.  Suppose for a
  contradiction that $v\neq 0$.  Then $\kappa(s,s)$ is collinear
  with $v$ and, by polarisation, so are $\kappa(s_1,s_2)$ for all
  $s_1,s_2\in S'$.  But this says that $V'$ is one-dimensional,
  contradicting the fact that $\dim (V')^\perp \leq 1$.
\end{proof}

\bibliographystyle{utphys}
\bibliography{4dKSA}

\end{document}